\documentclass{article}
\usepackage{nips15submit_e,times}

\usepackage{amsmath,amsfonts,graphpap,amscd,mathrsfs,graphicx,lscape}
\usepackage{epsfig,amssymb,amstext,xspace}
\usepackage{algorithm,algorithmic}
\usepackage{subfigure}

\usepackage{color}              
\usepackage{epsfig}

\usepackage{amsthm}
\usepackage{xspace}

\newcommand{\E}{\mathbb{E}}

\newtheorem{theorem}{Theorem}
\newtheorem{lemma}[theorem]{Lemma}

\newtheorem{corollary}[theorem]{Corollary}
\newtheorem{defn}[theorem]{Definition}

\newtheorem{proposition}[theorem]{Proposition}

\newenvironment{rtheorem}[3][]{
\bigskip
\noindent \ifthenelse{\equal{#1}{}}{\bf #2 #3}{\bf #2 #3 (#1)}
\begin{it}
}{\end{it}}

\def\squareforqed{\hbox{\rule{2.5mm}{2.5mm}}}

\def\QED{\ifmmode\squareforqed 
  \else{\nobreak\hfil   
    \penalty50                 
    \hskip1em                  
    \null                      
    \nobreak                   
    \hfil                      
    \squareforqed              
    \parfillskip=0pt           
    \finalhyphendemerits=0     
    \endgraf}                  
  \fi}

\def\blksquare{\rule{2mm}{2mm}}
\def\qedsymbol{\blksquare}
\newcommand{\bg}[1]{\medskip\noindent{\bf #1}}
\newcommand{\ed}{{\hfill\qedsymbol}\medskip}

\newenvironment{example}{\bg{Example. }}{\ed}

\newcommand{\R}{\ensuremath{\mathbb R}}

\newcommand{\A}{\ensuremath{\mathcal{A}}}

\newcommand{\comment}[1]{}
 {}

\newcommand{\junk}[1]{}

\newcommand{\mP}{\ensuremath{\mathcal{P}}}

\newlength{\tmp} \newlength{\lpsx} \newlength{\lpsy} \newlength{\upsx} \newlength{\upsy}

\newcommand{\poa}{\text{\textsc{PoA}} }

\newcommand{\opt}{\text{\textsc{Opt}} }

\renewcommand{\vec}[1]{\ensuremath{{\bf #1}}}

\newcommand{\st}{\ensuremath{s}}
\newcommand{\mst}{\ensuremath{w}}

\newcommand{\dotp}[2]{\ensuremath{\left\langle #1, #2 \right\rangle}}

\newcommand{\mR}{\ensuremath{\mathcal{R}}}
\newcommand{\myprop}{\ensuremath{\texttt{RVU}}}
\newcommand{\stable}{\ensuremath{\kappa}}
\newcommand{\knob}{\ensuremath{\rho}}

\DeclareMathOperator*{\argmin}{argmin}
\DeclareMathOperator*{\argmax}{argmax}

\nipsfinalcopy 

\begin{document}
\title{Fast Convergence of Regularized Learning in Games}
\author{
Vasilis Syrgkanis \\ Microsoft Research \\ New York, NY \\ \texttt{vasy@microsoft.com}
\And
Alekh Agarwal \\ Microsoft Research \\ New York, NY \\ \texttt{alekha@microsoft.com}
\And
Haipeng Luo \\ Princeton University \\ Princeton, NJ \\ \texttt{haipengl@cs.princeton.edu}
\And
Robert E. Schapire \\ Microsoft Research \\ New York, NY \\ \texttt{schapire@microsoft.com}
}
\maketitle

\begin{abstract}

We show that natural classes of regularized learning algorithms with a
form of recency bias achieve faster convergence rates to approximate
efficiency and to coarse correlated equilibria in multiplayer normal form
games. When each player in a game uses an algorithm from our class,
their individual regret decays at $O(T^{-3/4})$, while the sum of
utilities converges to an approximate optimum at $O(T^{-1})$--an
improvement upon the worst case $O(T^{-1/2})$ rates. We show a
black-box reduction for any algorithm in the class to achieve
$\tilde{O}(T^{-1/2})$ rates against an adversary, while maintaining the faster
rates against algorithms in the class. Our results extend those of
Rakhlin and Shridharan~\cite{Rakhlin2013} and Daskalakis et
al.~\cite{Daskalakis2014}, who only analyzed two-player zero-sum games
for specific algorithms.

\end{abstract}

\section{Introduction}

What happens when players in a game interact with one another, all of them
acting independently and selfishly to maximize their own utilities?
If they are smart, we intuitively expect their utilities
--- both individually and as a group --- to grow, perhaps even to
approach the best possible.
We also expect the dynamics of their behavior to eventually reach some kind of
equilibrium.
Understanding these dynamics is central to game theory as well as its
various application areas, including economics, network routing,
auction design, and evolutionary biology.

It is natural in this setting for the players to each make use of a
no-regret learning algorithm for making their decisions, an approach
known as \emph{decentralized no-regret dynamics}.  No-regret
algorithms are a strong match for playing games because their regret
bounds hold even in adversarial environments.  As a benefit, these
bounds ensure that each player's utility approaches optimality. When
played against one another, it can also be shown that the sum of
utilities approaches an approximate optimum~\cite{Blum2008,
Roughgarden2009}, and the player strategies converge to an equilibrium
under appropriate conditions~\cite{Foster1997, Blum2007, Freund1999},
at rates governed by the regret bounds.  Well-known families of
no-regret algorithms include
multiplicative-weights~\cite{Littlestone1994,Freund1997}, Mirror
Descent~\cite{Nemirovsky1983}, and Follow the Regularized/Perturbed
Leader~\cite{Kalai2005}.  (See \cite{Cesa-Bianchi2006,Shwartz2012} for
excellent overviews.)  For all of these, the average regret vanishes
at the {worst-case} rate of $O(1/\sqrt{T})$, which is unimprovable in
fully {adversarial} scenarios.

However, the players in our setting are facing other similar,
predictable no-regret learning algorithms, a chink that hints at the
possibility of improved convergence rates for such dynamics.  This was
first observed and exploited by
Daskalakis~et~al.~\cite{Daskalakis2014}.  For two-player zero-sum
games, they developed a decentralized variant of Nesterov's
accelerated saddle point algorithm~\cite{Nesterov2005} and showed that
each player's average regret converges at the remarkable rate of
$O(1/T)$.  Although the resulting dynamics are somewhat unnatural, in
later work, Rakhlin and Sridharan~\cite{Rakhlin2013} showed
surprisingly that the same convergence rate holds for a simple variant
of Mirror Descent with the seemingly minor modification that the last
utility observation is counted twice.

Although major steps forward, both these works are limited to 
two-player zero-sum games, the very simplest case.
As such, they do not cover many practically important settings,
such as auctions or routing games,
which are decidedly not zero-sum, and which involve many
independent actors.

In this paper, we vastly generalize these techniques to the
practically important but far more challenging case of arbitrary
multi-player normal-form games, giving natural no-regret dynamics
whose convergence rates are much faster than previously possible for
this general setting.  

\paragraph{Contributions.}
We show that the average welfare of the game, that is, the sum of
player utilities, converges to approximately optimal welfare at the
rate $O(1/T)$, rather than the previously known rate of
$O(1/\sqrt{T})$. Concretely, we show a natural class of regularized
no-regret algorithms with recency bias that achieve welfare at least
$({\lambda}/({1+\mu}))\opt - O({1}/{T})$, where $\lambda$ and $\mu$
are parameters in a smoothness condition on the game introduced by
Roughgarden~\cite{Roughgarden2009}. For the same class of algorithms,
we show that each individual player's average regret converges to zero
at the rate $O\left(T^{-3/4}\right)$. 
Thus, our results entail an algorithm for computing coarse correlated equilibria in a decentralized
manner with significantly faster convergence than existing methods.

We additionally give a black-box reduction that preserves the fast
rates in favorable environments, while robustly maintaining
$\tilde{O}({1}/{\sqrt{T}})$ regret against any opponent in the worst case.

Even for two-person zero-sum games, our results for general games
expose a hidden generality and modularity underlying the previous
results~\cite{Daskalakis2014,Rakhlin2013}.  First, our analysis
identifies stability and recency bias as key structural ingredients of
an algorithm with fast rates. This covers the Optimistic Mirror
Descent of Rakhlin and Sridharan~\cite{Rakhlin2013} as an example, but
also applies to optimistic variants of Follow the Regularized Leader
(FTRL), including dependence on arbitrary weighted windows in the
history as opposed to just the utility from the last round. Recency bias
is a behavioral pattern commonly observed in game-theoretic
environments~\cite{Fudenberg2014}; as such, our results can be viewed
as a partial theoretical justification. Second, previous approaches
in~\cite{Daskalakis2014,Rakhlin2013} on achieving both faster
convergence against similar algorithms while at the same time
$\tilde{O}(1/\sqrt{T})$ regret rates against adversaries were shown via ad-hoc
modifications of specific algorithms. We give a black-box modification
which is not algorithm specific and works for all these optimistic
algorithms.

Finally, we simulate a 4-bidder simultaneous auction game, and compare
our optimistic algorithms against Hedge~\cite{Freund1997} in terms of
utilities, regrets and convergence to equilibria.

\section{Repeated Game Model and Dynamics}
\label{sec:prelims}
Consider a static game $G$ among a set $N$ of $n$ players. Each player
$i$ has a strategy space $S_i$ and a utility function $u_i:
S_1\times \ldots \times S_n \rightarrow [0,1]$ that maps a strategy
profile $\vec{\st}=(\st_1,\ldots,\st_n)$ to a utility
$u_i(\vec{\st})$. We assume that the strategy space of each player is
finite and has cardinality $d$, i.e. $|S_i|=d$.  We denote with
$\vec{w}=(\vec{w}_1,\ldots, \vec{w}_n)$ a profile of mixed strategies,
where $\vec{w}_i \in \Delta(S_i)$ and $\mst_{i,x}$ is the probability
of strategy $x\in S_i$. Finally let $U_i(\vec{\mst})
= \E_{\vec{\st}\sim \vec{\mst}}[u_i(\vec{\st})]$, the expected utility
of player $i$.

We consider the setting where the game $G$ is played repeatedly for
$T$ time steps. At each time step $t$ each player $i$ picks a mixed
strategy $\vec{\mst}_i^t\in \Delta(S_i)$. At the end of the iteration
each player $i$ observes the expected utility he would have received
had he played any possible strategy $x\in S_i$. More formally, let
$u_{i,x}^t
= \E_{\vec{s}_{-i}\sim \vec{w}_{-i}^t}[u_i(x,\vec{s}_{-i})]$, where
$\vec{s}_{-i}$ is the set of strategies of all but the $i^{th}$ player,
and let $\vec{u}_i^t = (u_{i,x}^t)_{x\in S_i}$. At the end of each
iteration each player $i$ observes $\vec{u}_i^t$. Observe that the
expected utility of a player at iteration $t$ is simply the inner
product $\dotp{ \vec{\mst}_i^t}{\vec{u}_i^t}$.

\paragraph{No-regret dynamics.} We assume that the players each decide
their strategy $\vec{\mst}_i^t$ based on a vanishing regret
algorithm. Formally, for each player $i$, the regret after $T$
time steps is equal to the maximum gain he could have achieved by
switching to any other fixed strategy:
\begin{equation*}
r_i(T)
= \sup_{\vec{w}_i^*\in \Delta(S_i)} \sum_{t=1}^{T} \dotp{\vec{w}_i^*-\vec{w}_i^t}{ \vec{u}_i^t}. 
\end{equation*}
The algorithm has vanishing regret if $r_i(T) = o(T)$.

\paragraph{Approximate Efficiency of No-Regret Dynamics.}

We are interested in analyzing the average welfare of such vanishing
regret sequences. For a given strategy profile $\vec{\st}$ the social
welfare is defined as the sum of the player utilities: $W(\vec{s})
= \sum_{i\in N} u_i(\vec{s})$. We overload notation to denote
$W(\vec{\mst}) = \E_{\vec{\st}\sim \vec{\mst}}[W(\vec{\st})]$. We want
to lower bound how far the average welfare of the sequence is, with
respect to the optimal welfare of the static game:
\begin{equation*}
\opt = \max_{\vec{\st}\in S_1\times \ldots \times S_n} W(\vec{s}).
\end{equation*}
This is the optimal welfare achievable in the absence of player
incentives and if a central coordinator could dictate each player's
strategy. We next define a class of games first identified by
Roughgarden~\cite{Roughgarden2009} on which we can approximate the
optimal welfare using decoupled no-regret dynamics.


\begin{defn}[Smooth game~\cite{Roughgarden2009}] A game is
$(\lambda,\mu)$-smooth if there exists a strategy profile
$\vec{\st}^*$ such that for any strategy profile $\vec{\st}$:
$\sum_{i\in N} u_i(\st_i^*,\vec{\st}_{-i}) \geq \lambda \opt - \mu
W(\vec{\st})$.
\label{defn:smoothness}
\end{defn}

In words, any player using his optimal strategy continues to do well
irrespective of other players' strategies. This condition directly
implies near-optimality of no-regret dynamics as we show below.

\begin{proposition}
In a $(\lambda,\mu)$-smooth game, if each player $i$ suffers regret at
most $r_i(T)$, then:
\begin{equation*}
\frac{1}{T} \sum_{t=1}^{T}
 W(\vec{\mst}^t) \geq \frac{\lambda}{1+\mu} \opt
 - \frac{1}{1+\mu}\frac{1}{T}\sum_{i\in N} r_i(T)
 = \frac{1}{\rho} \opt - \frac{1}{1+\mu}\frac{1}{T}\sum_{i\in N}
 r_i(T),
\end{equation*}
 \label{prop:smoothness}
where the factor $\rho=(1+\mu)/\lambda$
 is called the \emph{price of anarchy}
(\poa \hspace{-3pt}).  
\end{proposition}
This proposition is essentially a more explicit version of
Roughgarden's result~\cite{Roughgarden2009};
we provide a proof in the appendix for
completeness. The result shows that the convergence to \poa is driven
by the quantity $\frac{1}{1+\mu} \frac{1}{T}\sum_{i\in N}
r_i(T)$. There are many algorithms which achieve a regret rate of
\mbox{$r_i(T) = O(\sqrt{\log(d) T})$}, in which case the latter theorem would
imply that the average welfare converges to \poa at a rate of
$O(n\sqrt{\log(d)/T})$. As we will show, for some natural classes of
no-regret algorithms the average welfare converges at the much faster
rate of $O(n^2 \log(d)/T)$.

\section{Fast Convergence to Approximate Efficiency} 
\label{sec:fast}

In this section, we present our main theoretical results
characterizing a class of no-regret dynamics which lead to faster
convergence in smooth games. We begin by describing this class.

\begin{defn}[\myprop~property]
  We say that a vanishing regret algorithm satisfies the Regret
  bounded by Variation in Utilities (\myprop) property with parameters
  $\alpha > 0$ and $0 < \beta \leq \gamma$ and a pair of dual norms
  $(\|\cdot\|, \|\cdot\|_*)$\footnote{The dual to a norm
    $\|\cdot\|$ is defined as $\|v\|_* = \sup_{\|u\| \leq 1}
    \dotp{u}{v}$.}  if its regret on any sequence of utilities
  $\vec{u}^1, \vec{u}^2, \ldots, \vec{u}^T$ is bounded as
  \begin{equation}
    \sum_{t=1}^{T} \dotp{\vec{w}^*- \vec{w}^t}{\vec{u}^t} \leq \alpha
    +\beta \sum_{t=1}^{T} \| \vec{u}^t - \vec{u}^{t-1}\|_*^2 -
    \gamma\sum_{t=1}^{T} \|\vec{w}^{t} - \vec{w}^{t-1}\|^2.
    \label{eqn:reg-form}
  \end{equation}  
  \label{defn:alg-class}
\end{defn}

Typical online learning algorithms such as Mirror Descent and FTRL do
not satisfy the \myprop~property in their vanilla form, as the middle
term grows as $\sum_{t=1}^T \|\vec{u}^t\|_*^2$ for these
methods. However, Rakhlin and Sridharan~\cite{Rakhlin2013a} give a modification
of Mirror Descent with this property, and we will present a similar
variant of FTRL in the sequel.

We now present two sets of results
when each player uses an algorithm with this property. The first
discusses the convergence of social welfare, while the
second governs the convergence of the individual players' utilities at
a fast rate.

\subsection{Fast Convergence of Social Welfare}

Given Proposition~\ref{prop:smoothness}, we only need to understand
the evolution of the sum of players' regrets $\sum_{t=1}^T r_i(T)$ in
order to obtain convergence rates of the social welfare. Our main
result in this section bounds this sum when each player uses dynamics
with the \myprop~property.

\begin{theorem}\label{thm:sufficient}
Suppose that the algorithm of each player $i$ satisfies the property
\myprop~with parameters $\alpha, \beta$ and $\gamma$ such that
$\beta\leq \gamma/(n-1)^2$ and $\|\cdot\| = \|\cdot\|_1$. Then
$\sum_{i\in N} r_i(T) \leq \alpha n$.
\end{theorem}
\begin{proof}
Since $u_i(\vec{\st})\leq 1$, definitions imply:
$\|\vec{u}_i^t-\vec{u}_i^{t-1}\|_* \leq \sum_{\vec{s}_{-i}} \left|
\prod_{j\neq i} \mst_{j,s_j}^t - \prod_{j\neq
  i}\mst_{j,s_j}^{t-1}\right|.$
The latter is the total variation distance of two product
distributions. By known properties of total variation (see
e.g.~\cite{Hoeffding1958}), this is bounded by the sum of the total
variations of each marginal distribution:
\begin{equation}
\sum_{\vec{s}_{-i}} \left| \prod_{j\neq i} \mst_{j,s_j}^t -
\prod_{j\neq i}\mst_{j,s_j}^{t-1}\right|\leq \sum_{j\neq i} \|
\vec{\mst}_j^t-\vec{\mst}_j^{t-1}\|
\end{equation}
By Jensen's inequality, $\left(\sum_{j\neq i} \| \vec{\mst}_j^t
-\vec{\mst}_{j}^{t-1}\|\right)^2 \leq (n-1) \sum_{j\neq i} \|
\vec{\mst}_j^t- \vec{\mst}_j^{t-1}\|^2$, so that
\begin{equation*}
\sum_{i\in N} \|\vec{u}_i^t-\vec{u}_i^{t-1}\|_*^2 \leq (n-1)
\sum_{i\in N} \sum_{j\neq i} \| \vec{\mst}_j^t-
\vec{\mst}_j^{t-1}\|^2 = (n-1)^2 \sum_{i\in N} \| \vec{\mst}_i^t-
\vec{\mst}_i^{t-1}\|^2.
\end{equation*}
The theorem follows by summing up the
\myprop~property~\eqref{eqn:reg-form} for each player $i$ and
observing that the summation of the second terms is smaller than that
of the third terms and thereby can be dropped.
\end{proof}
\textbf{Remark:} The rates from the theorem depend on $\alpha$, which
will be $O(1)$ in the sequel. The above theorem
extends to the case where $\|\cdot\|$ is any norm equivalent to the
$\ell_1$ norm. The resulting requirement on $\beta$ in terms of
$\gamma$ can however be more stringent. Also, the theorem does not
require that all players use the same no-regret algorithm unlike
previous results~\cite{Daskalakis2014, Rakhlin2013}, as long as each
player's algorithm satisfies the \myprop~property with a common bound
on the constants.

We now instantiate the result with examples that satisfy the
\myprop~property with different constants.

\subsubsection{Optimistic Mirror Descent}

The optimistic mirror descent (OMD) algorithm
of Rakhlin and Sridharan~\cite{Rakhlin2013a}
is parameterized by an
adaptive predictor sequence $\vec{M}_i^t$ and a
regularizer\footnote{Here and in the sequel, we can use a different
regularizer $\mR_i$ for each player $i$, without qualitatively
affecting any of the results.} $\mR$ which is
$1$-strongly convex\footnote{$\mR$ is 1-strongly convex if
    $\mR\left(\frac{u+v}{2}\right) \leq \frac{\mR(u) + \mR(v)}{2}
    - \frac{\|u-v\|^2}{8}$, $\forall u,v$.
}
with respect to a norm $\|\cdot\|$.
Let $D_{\mR}$ denote the Bregman divergence associated with $\mR$.
Then the update rule
is defined as follows: let
$\vec{g}_i^0=\argmin_{\vec{g}\in \Delta(S_i)} \mR(\vec{g})$ and 
\[\Phi(\vec{u},\vec{g})=\argmax_{\vec{w}\in \Delta(S_i)}\eta \cdot \dotp{\vec{w}}{\vec{u}}-D_{\mR}(\vec{w},\vec{g}),\] 
then:
\begin{align*}
\vec{\mst}_i^t
= \Phi(\vec{M}_i^t,\vec{g}_i^{t-1}), ~~ \mbox{and}~~
\vec{g}_i^{t}
= \Phi(\vec{u}_i^t,\vec{g}_i^{t-1}) 
\end{align*}

Then the following proposition can be obtained for this method.

\begin{proposition}
  The OMD algorithm using stepsize $\eta$ and $\vec{M}_i^t
  = \vec{u}_i^{t-1}$satisfies the \myprop~property with constants
  $\alpha = R/\eta$, $\beta = \eta$, $\gamma = 1/(8\eta)$, where $R
  = \max_i\sup_{f} D_{\mR}(f,\vec{g}_i^0)$.
\label{prop:omd}
\end{proposition}
The proposition follows by further crystallizing the arguments of
Rakhlin and Sridaran \cite{Rakhlin2013}, and we provide a proof in the
appendix for completeness. The above proposition, along with
Theorem~\ref{thm:sufficient}, immediately yields the following corollary,
which had been proved by Rakhlin and Sridharan~\cite{Rakhlin2013}
for two-person zero-sum games, and which we here extend to general games.

\begin{corollary} \label{cor:fast-omd}
If each player runs OMD with $\vec{M}_i^t= \vec{u}_i^{t-1}$ and
stepsize $\eta= 1/(\sqrt{8} (n-1))$, then we have $\sum_{i\in N}
r_i(T) \leq n R/\eta \leq n(n-1)\sqrt{8} R= O(1)$.
\end{corollary}
The corollary follows by noting that the condition
$\beta \leq \gamma/(n-1)^2$ is met with our choice of $\eta$.

\subsubsection{Optimistic Follow the Regularized Leader}

We next consider a different class of algorithms denoted
as \emph{optimistic follow the regularized leader} (\emph{OFTRL}).
This algorithm is similar but not equivalent to OMD, and is an
analogous extension of standard FTRL~\cite{Kalai2005}.  This algorithm
takes the same parameters as for OMD and is defined as follows: Let
$\vec{\mst}_i^0
= \argmin_{\vec{\mst}\in \Delta(S_i)} \mR(\vec{\mst})$ and:
\begin{equation*}
 \vec{\mst}_i^T
 = \argmax_{\vec{\mst} \in \Delta(S_i)} \dotp{\vec{\mst}}{\sum_{t=1}^{T-1} \vec{u}_i^t
 + \vec{M}_i^T}-\frac{\mR(\vec{\mst})}{\eta}.
\end{equation*}

We consider three variants of OFTRL with different choices of the
sequence $\vec{M}_i^t$, incorporating the recency bias in different
forms.

\paragraph{One-step recency bias:} 
The simplest
form of OFTRL uses $\vec{M}_i^t = \vec{u}_i^{t-1}$ and obtains the
following result, where
$R = \max_i \left(\sup_{\vec{f}\in \Delta(S_i)} \mR(\vec{f})-\inf_{\vec{f}\in \Delta(S_i)}\mR(\vec{f})\right)$.

\begin{proposition}
  The OFTRL algorithm using stepsize $\eta$ and $\vec{M}_i^t
  = \vec{u}_i^{t-1}$ satisfies the \myprop~property with constants
  $\alpha = R/\eta$, $\beta = \eta$ and $\gamma
  = 1/(4\eta).$
\label{prop:oftrl-onestep}
\end{proposition}

Combined with Theorem~\ref{thm:sufficient}, this yields the following
constant bound on the total regret of all players:

\begin{corollary}\label{cor:fast-oftrl}
If each player runs OFTRL with $\vec{M}_i^t= \vec{u}_i^{t-1}$ and
$\eta= 1/(2 (n-1))$, then we have $\sum_{i\in N} r_i(T) \leq n
R/\eta \leq 2n(n-1) R= O(1).$
\end{corollary}

Rakhlin and Sridharan~\cite{Rakhlin2013a} also analyze an FTRL
variant, but require a self-concordant barrier for the constraint set
as opposed to an arbitrary strongly convex regularizer, and their
bound is missing the crucial negative terms of the \myprop~property
which are essential for obtaining Theorem~\ref{thm:sufficient}.

\paragraph{$H$-step recency bias:} More generally, given a window size
$H$, one can define $\vec{M}_i^t = \sum_{\tau =
t-H}^{t-1} \vec{u}_i^\tau/H$. We have the following proposition. 

\begin{proposition}
  The OFTRL algorithm using stepsize $\eta$ and $\vec{M}_i^t
  = \sum_{\tau = t-H}^{t-1} \vec{u}_i^\tau/H$ satisfies
  the \myprop~property with constants $\alpha = R/\eta$, $\beta = \eta
  H^2$ and $\gamma = 1/(4\eta).$
\label{prop:oftrl-Hstep}
\end{proposition}
Setting $\eta = 1/(2H(n-1))$, we obtain the analogue of
Corollary~\ref{cor:fast-oftrl}, with an extra factor of $H$.

\paragraph{Geometrically discounted recency bias:} The next
proposition considers an alternative form of recency bias which
includes all the previous utilities, but with a geometric
discounting. 

\begin{proposition}
  The OFTRL algorithm using stepsize $\eta$ and $\vec{M}_i^t
  = \frac{1}{\sum_{\tau = 0}^{t-1} \delta^{-\tau}} \sum_{\tau =
  0}^{t-1} \delta^{-\tau} \vec{u}_i^\tau$ satisfies
  the \myprop~property with constants $\alpha = R/\eta$, $\beta
  = \eta/(1- \delta)^3$ and $\gamma = 1/(8\eta).$
\label{prop:oftrl-discount}
\end{proposition}
Note that these choices for $\vec{M}_i^t$ can also be
used in OMD with qualitatively similar results.


\subsection{Fast Convergence of Individual Utilities}\label{sec:equilibrium}

The previous section shows implications of the \myprop~property on the
social welfare. This section complements these with a similar result for
each player's individual utility.

\begin{theorem}\label{thm:sufficient-2}
Suppose that the players use algorithms satisfying the
\myprop~property with parameters $\alpha > 0, \beta > 0,\gamma \geq
0$. If we further have the stability property $\|\vec{\mst}_i^t -
\vec{\mst}_i^{t+1}\|\leq \stable$, then for any player
$\sum_{t=1}^{T} \dotp{\vec{\mst}_i^*-\vec{\mst}_i^t}{ \vec{u}_i^t}\leq
\alpha + \beta \stable^2 (n-1)^2 T.$
\end{theorem}
Similar reasoning as in Theorem \ref{thm:sufficient} yields:
\mbox{$\| \vec{u}_i^t - \vec{u}_i^{t-1}\|_*^2 \leq (n-1) \sum_{j\neq i}
 \|\vec{\mst}_j^t- \vec{\mst}_j^{t-1}\|^2\leq (n-1)^2 \stable^2$}, and
 summing the terms gives the theorem.

Noting that OFTRL satisfies the \myprop~property with constants given
in Proposition~\ref{prop:oftrl-onestep} and stability property with
$\stable = 2\eta$ (see Lemma~\ref{lem:oftrl-stability} in the
appendix), we have the following corollary.

\begin{corollary}\label{cor:oftrl-bound}
If all players use the OFTRL algorithm with $\vec{M}_i^t =
\vec{u}_i^{t-1}$ and $\eta = (n-1)^{-1/2}T^{-1/4}$, then we have
$\sum_{t=1}^{T} \dotp{\vec{\mst}_i^*-\vec{\mst}_i^t}{ \vec{u}_i^t}\leq
(R+4)\sqrt{n-1}\cdot T^{1/4}.$
\end{corollary}

Similar results hold for the other forms of recency bias, as well as
for OMD. Corollary~\ref{cor:oftrl-bound} gives a fast convergence rate of
the players' strategies to the set of \emph{coarse correlated
equilibria} (CCE) of the game. This improves the previously known convergence rate 
$\sqrt{T}$ (e.g. \cite{Hart2000}) to CCE using natural, decoupled no-regret dynamics defined in \cite{Daskalakis2014}.


\section{Robustness to Adversarial Opponent}\label{sec:adversarial}

So far we have shown simple dynamics with rapid convergence properties
in favorable environments when each player in the game uses an
algorithm with the \myprop~property. It is natural to wonder if this
comes at the cost of worst-case guarantees when some players do not
use algorithms with this property. Rakhlin and
Sridharan~\cite{Rakhlin2013} address this concern by modifying the OMD
algorithm with additional smoothing and adaptive step-sizes so as to
preserve the fast rates in the favorable case while still
guaranteeing $O(1/\sqrt{T})$ regret for each player, no matter how
the opponents play. It is not so obvious how this modification might
extend to other procedures, and it seems undesirable to abandon the
black-box regret transformations we used to obtain
Theorem~\ref{thm:sufficient}. In this section, we present a
generic way of transforming an algorithm which satisfies
the \myprop~property so that it retains the fast convergence in
favorable settings, but always guarantees a worst-case regret of
$\tilde{O}(1/\sqrt{T})$.

In order to present our modification, we need a parametric form of
the \myprop~property which will also involve a tunable parameter of
the algorithm. For most online learning algorithms, this will
correspond to the step-size parameter used by the algorithm.

\begin{defn}[\myprop(\knob)~property]
  We say that a parametric algorithm $\A(\knob)$ satisfies the Regret
  bounded by Variation in Utilities$(\knob)$ ($\myprop(\knob)$)
  property with parameters $\alpha, \beta, \gamma > 0$ and a pair of
  dual norms $(\|\cdot\|, \|\cdot\|_*)$ if its regret on any sequence
  of utilities $\vec{u}^1, \vec{u}^2, \ldots, \vec{u}^T$ is bounded
  as \begin{equation} \sum_{t=1}^{T} \dotp{\vec{w}^*- \vec{w}^t}{\vec{u}^t} \leq \frac{\alpha}{\knob}
  +\knob\beta \sum_{t=1}^{T} \| \vec{u}^t - \vec{u}^{t-1}\|_*^2
  - \frac{\gamma}{\knob}\sum_{t=1}^{T} \|\vec{w}^{t}
  - \vec{w}^{t-1}\|^2.  \label{eqn:reg-form-param} \end{equation} \label{defn:alg-class-param}
\end{defn}

In both OMD and OFTRL algorithms from Section~\ref{sec:fast}, the
parameter $\knob$ is precisely the stepsize $\eta$. We now show an
adaptive choice of $\knob$ according to an epoch-based doubling
schedule.

\paragraph{Black-box reduction.} 
Given a parametric algorithm $\A(\knob)$ as a black-box we construct a
wrapper $\A'$ based on the doubling trick: The algorithm of each
player proceeds in epochs. At each epoch $r$ the player $i$ has an
upper bound of $B_r$ on the quantity
$\sum_{t=1}^{T} \|\vec{u}_i^{t}-\vec{u}_i^{t-1}\|_*^2$. We start with
a parameter $\eta_*$ and $B_1=1$, and for $\tau = 1, 2, \ldots, T$
repeat:
\begin{enumerate}
\item Play according to $\A(\eta_r)$ and receive $\vec{u}_i^\tau$.
\item If $\sum_{t=1}^{\tau}|\vec{u}_i^{t}-\vec{u}_i^{t-1}\|_*^2\geq
B_r$:
\begin{enumerate}
\item Update $r \leftarrow r+1$, $B_r \leftarrow 2B_r$, $\eta_r
= \min\left\{\frac{\alpha}{\sqrt{B_r}}, \eta_*\right\}$,
with $\alpha$
as in Equation~\eqref{eqn:reg-form-param}.
\item  Start a new run of $\A$ with parameter $\eta_r$.
\end{enumerate}
\end{enumerate}

\begin{theorem}
Algorithm $\A'$ achieves regret at most the minimum of the following
two terms:
\begin{small}
\begin{align}
\sum_{t=1}^{T} \dotp{\vec{\mst}_i^*-\vec{\mst}_i^t}{\vec{u}_i^t} \leq& \log(T)\left(2+\frac{\alpha}{\eta_*} + (2+\eta_*\cdot \beta)\sum_{t=1}^T \|\vec{u}_i^t - \vec{u}_i^{t-1}\|_*^2\right) - \frac{\gamma}{\eta_*}\sum_{t=1}^{T} \|\vec{\mst}_i^t-\vec{\mst}_i^{t-1}\|^2;
\label{eqn:nice}\\
\sum_{t=1}^{T} \dotp{\vec{\mst}_i^*-\vec{\mst}_i^t}{\vec{u}_i^t} \leq& \log(T)\left(1+\frac{\alpha}{\eta_*} + (1+\alpha\cdot\beta)\cdot\sqrt{2\sum_{t=1}^T \|\vec{u}_i^t - \vec{u}_i^{t-1}\|_* ^2}\right)
\label{eqn:adversarial}
\end{align}
\end{small}
\label{thm:adaptive}
\end{theorem}

That is, the algorithm satisfies the \myprop~property, and also has
regret that can never exceed $\tilde{O}(\sqrt{T})$.  The theorem thus
yields the following corollary, which illustrates the stated
robustness of $\A'$.

\begin{corollary}
Algorithm $\A'$, with $\eta_*=\frac{\gamma}{(2+\beta)(n-1)^2\log(T)}$,
achieves regret $\tilde{O}(\sqrt{T})$ against any adversarial
sequence, while at the same time satisfying the conditions of
Theorem \ref{thm:sufficient}. Thereby, if all players use such an
algorithm, then: $\sum_{i\in N} r_i(T) \leq
n \log(T)(\alpha/\eta_*+2)=\tilde{O}(1)$.
\end{corollary}
\begin{proof}
Observe that for such $\eta^*$, we have that:
$(2+\eta_*\cdot \beta)\log(T) \leq
(2+\beta)\log(T)\leq \frac{\gamma}{\eta_*(n-1)^2}$. Therefore,
algorithm $\A'$, satisfies the sufficient conditions of
Theorem \ref{thm:sufficient}.
\end{proof}

If $\A(\knob)$ is the OFTRL algorithm, then we know by
Proposition~\ref{prop:oftrl-onestep} that the above result applies
with $\alpha = R = \max_{\vec{w}} \mR(\vec{w})$, $\beta=1$,
$\gamma=\frac{1}{4}$ and $\knob = \eta$. Setting
$\eta_*=\frac{\gamma}{(2+\beta)(n-1)^2} = \frac{1}{12(n-1)^2}$, the
resulting algorithm $\A'$ will have regret at most:
$\tilde{O}(n^2 \sqrt{T})$ against an arbitrary adversary, while if
all players use algorithm $\A'$ then $\sum_{i\in N} r_i(T) =
O(n^3 \log(T))$.

An analogue of Theorem~\ref{thm:sufficient-2} can also be established
for this algorithm:

\begin{corollary}\label{cor:setting-individual-const}
If $\A$ satisfies the $\myprop(\knob)$~property, and also
$\|\vec{\mst}_i^t-\vec{\mst}_i^{t-1}\| \leq \stable\knob$, then $\A'$
with $\eta_*=T^{-1/4}$ achieves regret $\tilde{O}(T^{1/4})$ if played
against itself, and $\tilde{O}(\sqrt{T})$ against any opponent.
\end{corollary}

Once again, OFTRL satisfies the above conditions with $\stable = 2$,
implying robust convergence.

\section{Experimental Evaluation}

\begin{figure}[!t]
\centering
\subfigure{
\includegraphics[width=0.4\textwidth,height=0.3\textwidth]{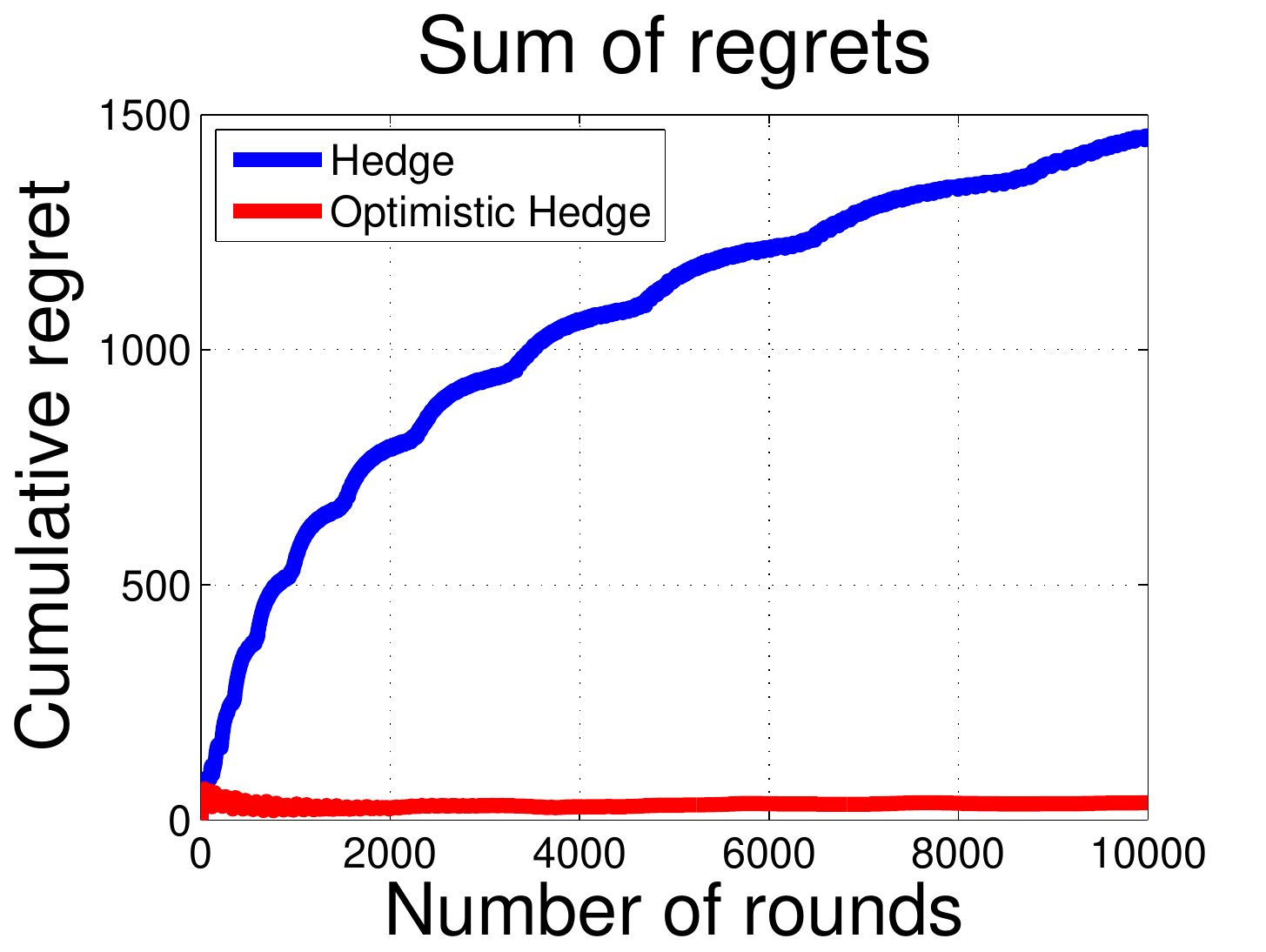}
}
\quad
\subfigure{
\includegraphics[width=0.4\textwidth,height=0.3\textwidth]{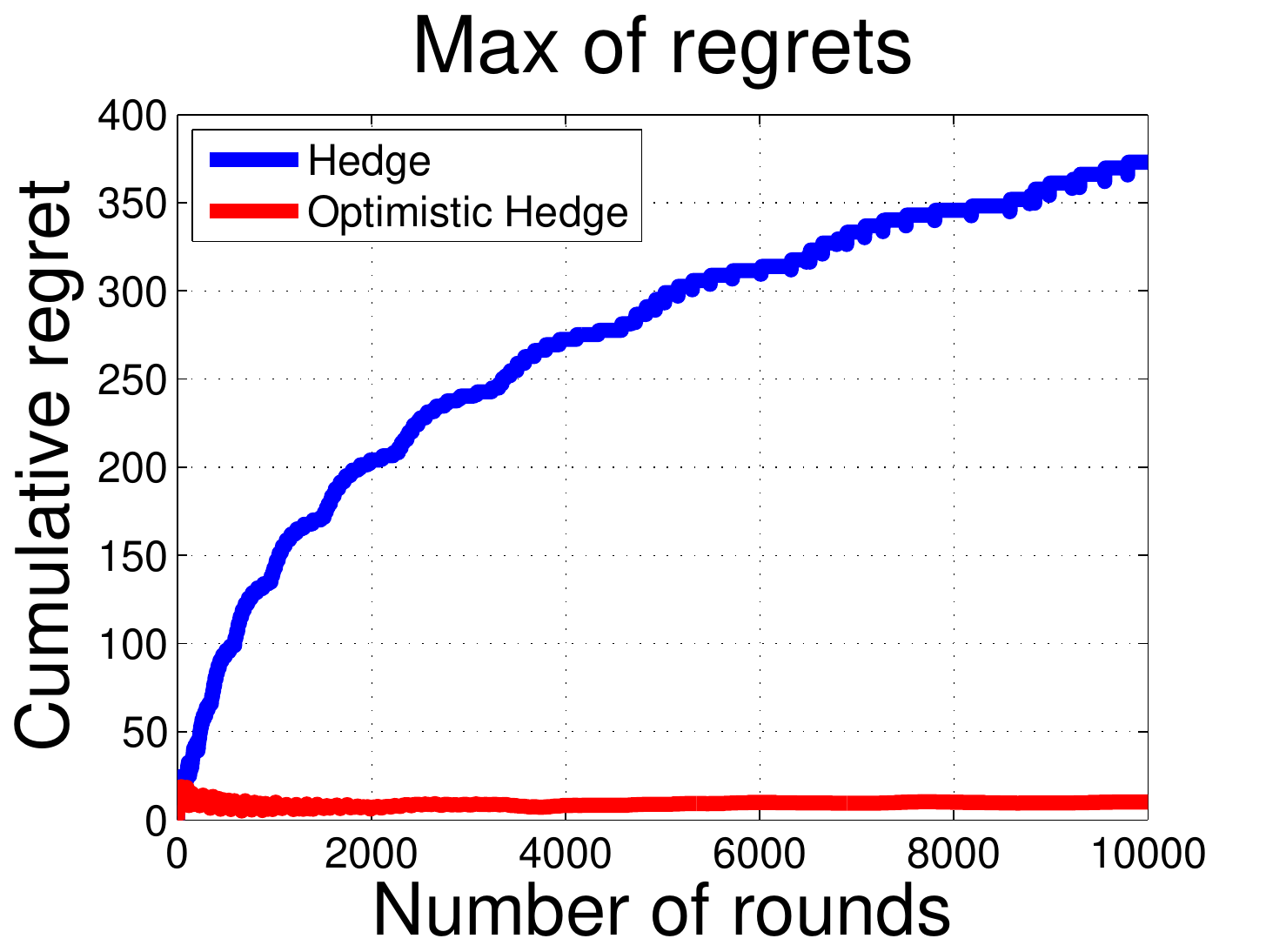}
}
\caption{Maximum and sum of individual regrets over time under the
  Hedge (blue) and \mbox{Optimistic Hedge} (red) dynamics.}\label{fig:regrets}
\end{figure}

We analyzed the performance of optimistic follow the regularized
leader with the entropy regularizer, which corresponds to the Hedge
algorithm~\cite{Freund1997} modified so that the last iteration's
utility for each strategy is double counted; we refer to it as
\emph{Optimistic Hedge}. More formally, the probability of player $i$
playing strategy $j$ at iteration $T$ is proportional to
$\exp\left(-\eta\cdot \left(\sum_{t=1}^{T-2}u_{ij}^t +
2u_{ij}^{T-1}\right)\right)$, rather than $\exp\left(-\eta\cdot
\sum_{t=1}^{T-1}u_{ij}^t\right)$ as is standard for Hedge.

We studied a simple auction where $n$ players are bidding for $m$
items. Each player has a value $v$ for getting at least one item and
no extra value for more items. The utility of a player is the value
for the allocation he derived minus the payment he has to make. The
game is defined as follows: simultaneously each player picks one of
the $m$ items and submits a bid on that item (we assume bids to be
discretized). For each item, the highest bidder wins and pays his
bid. We let players play this game repeatedly with each player
invoking either Hedge or optimistic Hedge. This game, and
generalizations of it, are known to be $(1-1/e,0)$-smooth
\cite{Syrgkanis2013}, if we also view the auctioneer as a player whose
utility is the revenue. The welfare of the game is the value of the
resulting allocation, hence not a constant-sum game. The welfare
maximization problem corresponds to the unweighted bipartite matching
problem. The $\poa$ captures how far from the optimal matching is the
average allocation of the dynamics. By smoothness we know it converges
to at least $1-1/e$ of the optimal.

\paragraph{Fast convergence of individual and average regret.} We run
the game for $n=4$ bidders and $m=4$ items and valuation $v=20$. The bids
are discretized to be any integer in $[1,20]$. We find that the sum of
the regrets and the maximum individual regret of each player are
remarkably lower under Optimistic Hedge as opposed to Hedge. In Figure
\ref{fig:regrets} we plot the maximum individual regret as well as the
sum of the regrets under the two algorithms, using $\eta = 0.1$ for
both methods.  Thus convergence to the set of coarse correlated
equilibria is substantially faster under Optimistic Hedge, confirming
our results in Section \ref{sec:equilibrium}. We also observe similar
behavior when each player only has value on a randomly picked
player-specific subset of items, or uses other step sizes.

\paragraph{More stable dynamics.} We observe that the behavior under Optimistic Hedge is more stable than under Hedge. In Figure \ref{fig:bids}, we plot the expected bid of a player on one of the items and his expected utility 
under the two dynamics. Hedge exhibits the sawtooth behavior that was observed in generalized first price auction run by Overture (see \cite[p.~21]{Edelman2005}). In stunning contrast, Optimistic Hedge leads to more stable expected bids over time. This stability property of optimistic Hedge is one of the main intuitive reasons for the fast convergence of its regret. 

\paragraph{Welfare.} In this class of games, we did not observe any significant difference
between the average welfare of the methods. The key reason is the
following: the proof that no-regret dynamics are approximately
efficient (Proposition~\ref{prop:smoothness}) only relies on the fact
that each player does not have regret against the strategy $s_i^*$
used in the definition of a smooth game. In this game, regret against
these strategies is experimentally comparable under both algorithms,
even though regret against the best fixed strategy is remarkably
different. This indicates a possibility for faster rates for Hedge in
terms of welfare. In Appendix \ref{sec:first-order}, we show fast
convergence of the efficiency of Hedge for cost-minimization games,
though with a worse \poa.

\begin{figure}[!t]
\centering
\subfigure{
\includegraphics[width=0.4\textwidth,height=0.3\textwidth]{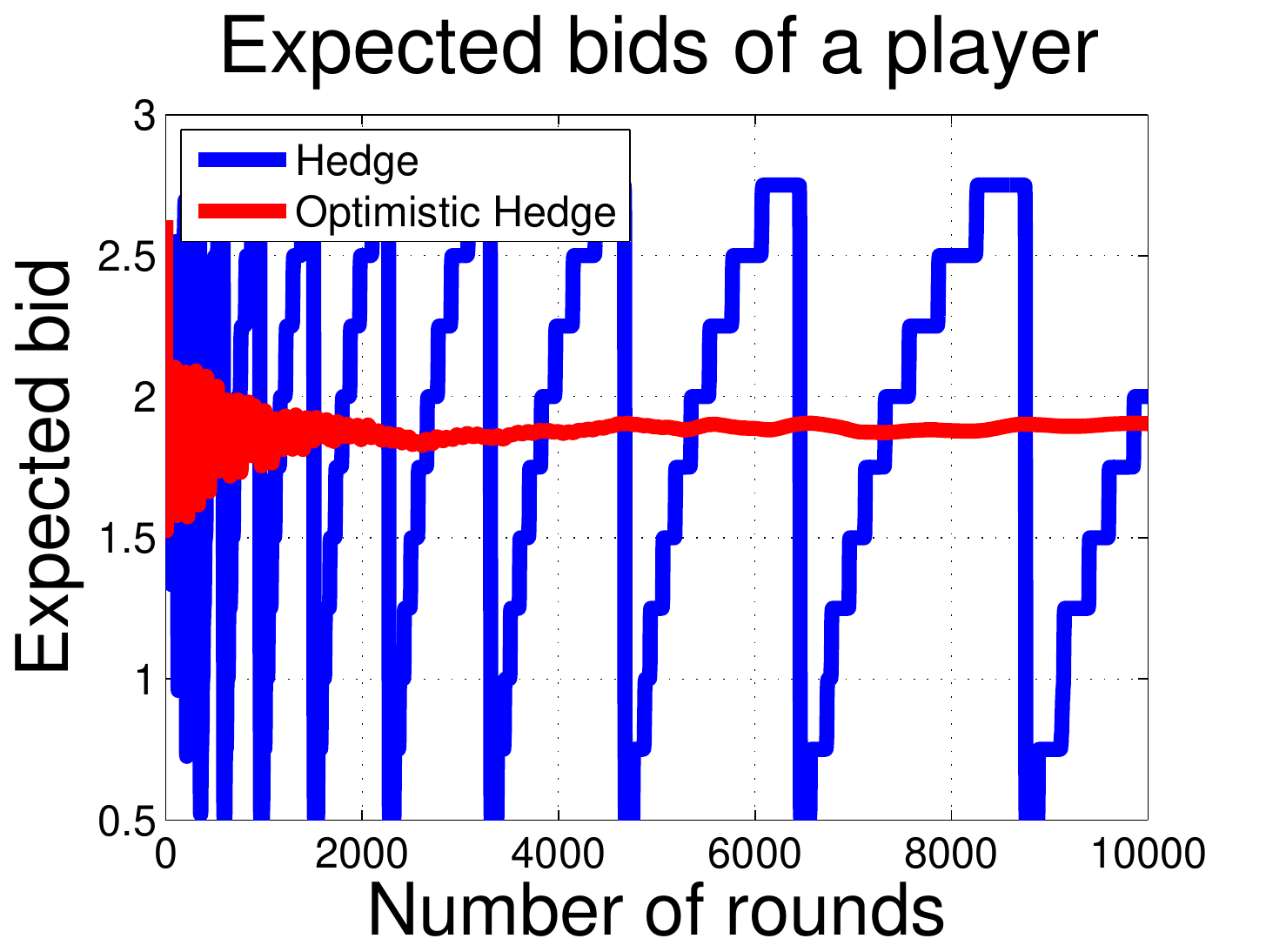}
}
\quad
\subfigure{
\includegraphics[width=0.4\textwidth,height=0.3\textwidth]{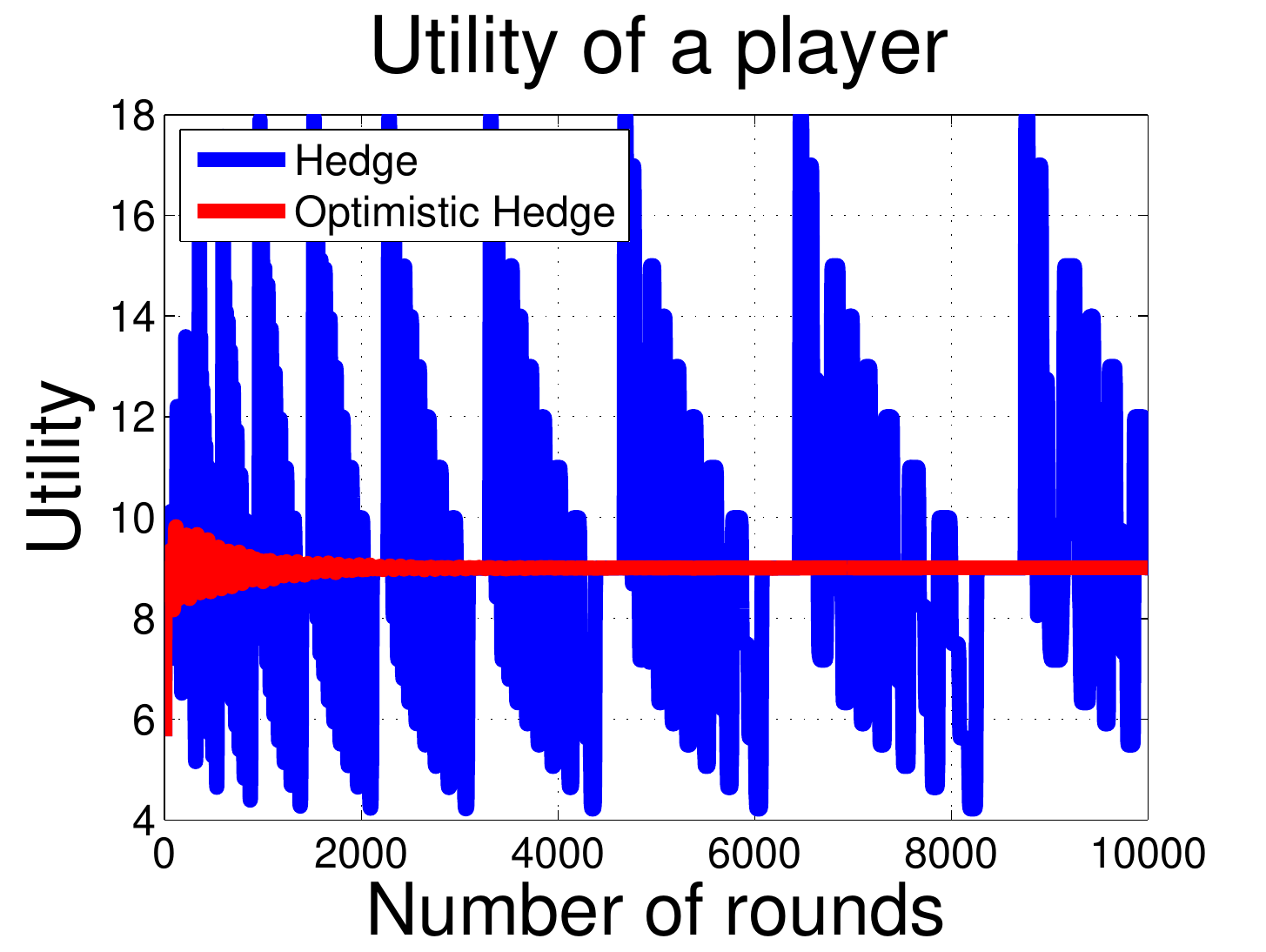}
}
\caption{Expected bid and per-iteration utility of a player on one of
  the four items over time, under Hedge (blue) and {Optimistic Hedge}
  (red) dynamics.}\label{fig:bids}
\end{figure}

\section{Discussion}
This work extends and generalizes a growing body of work on
decentralized no-regret dynamics in many ways. We
demonstrate a class of no-regret algorithms which enjoy rapid
convergence when played against each other, while being robust to
adversarial opponents. This has implications in computation of
correlated equilibria, as well as understanding the behavior of agents
in complex multi-player games. There are a number of interesting
questions and directions for future research which are suggested by
our results, including the following:

\textbf{Convergence rates for vanilla Hedge:} The fast rates of our
paper do not apply to algorithms such as Hedge without
  modification. Is this modification to satisfy \myprop~only
  sufficient or also necessary?  If not, are there counterexamples?
  In the supplement, we include a sketch hinting at such a
  counterexample, but also showing fast rates to a worse equilibrium
  than our optimistic algorithms.

\textbf{Convergence of players' strategies:} The OFTRL algorithm often
  produces much more stable trajectories empirically, as the players
  converge to an equilibrium, as opposed to say Hedge. A precise
  quantification of this desirable behavior would be of great
  interest.

\textbf{Better rates with partial information:} If the players
  do not observe the expected utility function, but only the moves of
  the other players at each round, can we still obtain faster rates?

\begin{small}
\bibliographystyle{plain}
\bibliography{bib-fast-convergence}
\end{small}
\newpage

\appendix
\setcounter{page}{1}


\begin{center}
\bf \Large Supplementary material for \\ ``Fast Convergence of Regularized Learning in Games''
\end{center}

\section{Proof of Proposition~\ref{prop:smoothness}}

\begin{rtheorem}{Proposition}{\ref{prop:smoothness}.}
In a $(\lambda,\mu)$-smooth game, if each player $i$ suffers regret at
most $r_i(T)$, then:
\begin{equation*}
\frac{1}{T} \sum_{t=1}^{T}
 W(\vec{\mst}^t) \geq \frac{\lambda}{1+\mu} \opt
 - \frac{1}{1+\mu}\frac{1}{T}\sum_{i\in N} r_i(T)
 = \frac{1}{\rho} \opt - \frac{1}{1+\mu}\frac{1}{T}\sum_{i\in N}
 r_i(T),
\end{equation*}
where the factor $\rho=(1+\mu)/\lambda$
 is called the \emph{price of total anarchy}
(\poa \hspace{-3pt}).  
\end{rtheorem}

\begin{proof}
Since each player $i$ has regret $r_i(T)$, we have that:
\begin{equation}
\sum_{t=1}^{T} \dotp{\vec{\mst}_i^t}{\vec{u}_i^t} \geq \sum_{t=1}^T u_{i,\st_i^*}^t - r_i(T)
\end{equation}
Summing over all players and using the smoothness property:
\begin{align*}
\sum_{t=1}^T W(\vec{\mst}^t) =~& \sum_{t=1}^T \sum_{i\in N}  \dotp{\vec{\mst}_i^t}{\vec{u}_i^t} \geq \sum_{t=1}^T \sum_{i\in N}u_{i,\st_i^*}^t - \sum_{i\in N} r_i(T)\\
 =~& \sum_{t=1}^{T} \E_{\vec{\st}\sim \vec{\mst^t}}\left[\sum_{i\in N} u_i(\st_i^*,\vec{\st}_{-i})\right] -\sum_{i\in N} r_i(T)\\
\geq~& \sum_{t=1}^{T}\left( \lambda \opt - \mu E_{\vec{\st}\sim \vec{\mst^t}}\left[W(\vec{\st})\right]\right) -\sum_{i\in N} r_i(T)\\
=~& \sum_{t=1}^{T}\left( \lambda \opt - \mu W(\vec{\mst}^t)\right) -\sum_{i\in N} r_i(T)
\end{align*}
By re-arranging we get the result.
\end{proof}

\section{Proof of Proposition~\ref{prop:omd}}

\begin{rtheorem}{Proposition}{\ref{prop:omd}.}
  The OMD algorithm using stepsize $\eta$ and $\vec{M}_i^t
  = \vec{u}_i^{t-1}$satisfies the \myprop~property with constants
  $\alpha = R/\eta$, $\beta = \eta$, $\gamma = 1/(8\eta)$, where $R
  = \max_i\sup_{f} D_{\mR}(f,\vec{g}_i^0)$.
\end{rtheorem}

We will use the following theorem of \cite{Rakhlin2013}.
\begin{theorem}[Raklin and Sridharan \cite{Rakhlin2013}]\label{thm:opt-omd-bound} The regret of a player under optimistic mirror descent and with respect to any $\vec{\mst}_i^*\in \Delta(S_i)$ is upper bounded by:
\begin{equation}
\sum_{t=1}^{T} \dotp{\vec{\mst}_i^*-\vec{\mst}_i^t}{\vec{u}_i^t} \leq \frac{R}{\eta} + \sum_{t=1}^T \|\vec{u}_i^t - \vec{M}_i^t\|_* \|\vec{\mst}_i^t - \vec{g}_i^t\| - \frac{1}{2\eta} \sum_{t=1}^{T} \left(\|\vec{\mst}_i^t-\vec{g}_i^t\|^2 + \|\vec{\mst}_i^t- \vec{g}_i^{t-1}\|^2\right)
\end{equation}
where $R =\sup_{f} D_{\mR}(f,g_0)$.
\end{theorem} 	
We show that if the players use optimistic mirror descent with $\vec{M}_i^t=\vec{u}_i^{t-1}$, then the regret of each player satisfies the sufficient condition presented in the previous section.  Some of the key facts (Equations \eqref{eqn:mult-split} and \eqref{eqn:squares}) that we use in the following proof appear in \cite{Rakhlin2013}. However, the formulation of the regret that we present in the following theorem is not immediately clear in their proof, so we present it here for clarity and completeness.

\begin{theorem}\label{thm:opt-omd-better-bound}
 The regret of a player under optimistic mirror descent with $\vec{M}_i^t=\vec{u}_i^{t-1}$ and with respect to any $\vec{\mst}_i^*\in \Delta(S_i)$ is upper bounded by:
\begin{equation}
\sum_{t=1}^{T} \dotp{\vec{\mst}_i^*-\vec{\mst}_i^t}{\vec{u}_i^t} \leq \frac{R}{\eta} + \eta\sum_{t=1}^T \|\vec{u}_i^t - \vec{u}_i^{t-1}\|_*^2  - \frac{1}{8\eta}\sum_{t=1}^{T} \|\vec{\mst}_i^t-\vec{\mst}_i^{t-1}\|^2
\end{equation}
\end{theorem} 	 
\begin{proof}
By Theorem \ref{thm:opt-omd-bound}, instantiated for $\vec{M}_i^t=\vec{u}_i^{t-1}$, we get:
\begin{multline*}
\sum_{t=1}^{T} \dotp{\vec{\mst}_i^*-\vec{\mst}_i^t}{\vec{u}_i^t} \leq \frac{R}{\eta} + \sum_{t=1}^T \|\vec{u}_i^t - \vec{u}_i^{t-1}\|_* \|\vec{\mst}_i^t - \vec{g}_i^t\|\\
 - \frac{1}{2\eta} \sum_{t=1}^{T} \left(\|\vec{\mst}_i^t-\vec{g}_i^t\|^2 + \|\vec{\mst}_i^t- \vec{g}_i^{t-1}\|^2\right)
\end{multline*}
Using the fact that for any $\rho>0$:
\begin{equation}\label{eqn:mult-split}
\|\vec{u}_i^t - \vec{M}_i^t\|_* \|\vec{\mst}_i^t - \vec{g}_i^t\| \leq \frac{\rho}{2} \|\vec{u}_i^t - \vec{M}_i^t\|_*^2 +\frac{1}{2\rho}\|\vec{\mst}_i^t - \vec{g}_i^t\|^2
\end{equation} 	
We get:
\begin{equation*}
\sum_{t=1}^{T} \dotp{\vec{\mst}_i^*-\vec{\mst}_i^t}{\vec{u}_i^t} \leq \frac{R}{\eta} + \frac{\rho}{2}\sum_{t=1}^T \|\vec{u}_i^t - \vec{u}_i^{t-1}\|_*^2  - \left(\frac{1}{2\eta}-\frac{1}{2\rho}\right) \sum_{t=1}^{T} \|\vec{\mst}_i^t-\vec{g}_i^t\|^2 - \frac{1}{2\eta}\sum_{t=1}^T \|\vec{\mst}_i^t- \vec{g}_i^{t-1}\|^2
\end{equation*}
For $\rho=2\eta$, the latter simplifies to:
\begin{align*}
\sum_{t=1}^{T} \dotp{\vec{\mst}_i^*-\vec{\mst}_i^t}{\vec{u}_i^t} \leq~& \frac{R}{\eta} +\eta\sum_{t=1}^T \|\vec{u}_i^t - \vec{u}_i^{t-1}\|_*^2  - \frac{1}{4\eta}\sum_{t=1}^{T} \|\vec{\mst}_i^t-\vec{g}_i^t\|^2 - \frac{1}{2\eta}\sum_{t=1}^T \|\vec{\mst}_i^t- \vec{g}_i^{t-1}\|^2\\
\leq~& \frac{R}{\eta} +\eta\sum_{t=1}^T \|\vec{u}_i^t - \vec{u}_i^{t-1}\|_*^2  - \frac{1}{4\eta}\sum_{t=1}^{T} \|\vec{\mst}_i^t-\vec{g}_i^t\|^2 - \frac{1}{4\eta}\sum_{t=1}^T \|\vec{\mst}_i^t- \vec{g}_i^{t-1}\|^2
\end{align*}
Last we use the fact that:
\begin{equation}\label{eqn:squares}
\|\vec{\mst}_i^t-\vec{\mst}_i^{t-1}\|^2 \leq 2\|\vec{\mst}_i^t-\vec{g}_i^{t-1}\|^2+2\|\vec{\mst}_i^{t-1}-\vec{g}_i^{t-1}\|^2
\end{equation}
Summing over all timesteps:
\begin{align*}
\sum_{t=1}^{T} \|\vec{\mst}_i^t-\vec{\mst}_i^{t-1}\|^2\leq~& 2\sum_{t=1}^{T}\|\vec{\mst}_i^t-\vec{g}_i^{t-1}\|^2 + 2\sum_{t=1}^{T}\|\vec{\mst}_i^{t-1}-\vec{g}_i^{t-1}\|^2\\
\leq~& 2\sum_{t=1}^{T}\|\vec{\mst}_i^t-\vec{g}_i^{t-1}\|^2 + 2\sum_{t=1}^{T}\|\vec{\mst}_i^t-g_i^{t}\|^2
\end{align*}
Dividing over by $\frac{1}{8\eta}$ and applying it in the previous upper bound on the regret, we get:
\begin{align*}
\sum_{t=1}^{T} \dotp{\vec{\mst}_i^*-\vec{\mst}_i^t}{\vec{u}_i^t} 
\leq~& \frac{R}{\eta} +\eta\sum_{t=1}^T \|\vec{u}_i^t - \vec{u}_i^{t-1}\|_*^2  - \frac{1}{8\eta}\sum_{t=1}^{T} \|\vec{\mst}_i^t-\vec{\mst}_i^{t-1}\|^2 
\end{align*}
\end{proof}

\section{Proof of Proposition~\ref{prop:oftrl-onestep}}

\begin{rtheorem}{Proposition}{\ref{prop:oftrl-onestep}.}
  The OFTRL algorithm using stepsize $\eta$ and $\vec{M}_i^t
  = \vec{u}_i^{t-1}$ satisfies the \myprop~property with constants
  $\alpha = R/\eta$, $\beta = \eta$ and $\gamma
  = 1/(4\eta).$
\end{rtheorem}

We first show that these algorithms achieve the same regret bounds as
optimistic mirror descent. This result does not appear in previous
work in any form.

Even though the algorithms do not make use of a secondary sequence, we
will still use in the analysis the notation:
\[\vec{g}_i^T
= \argmax_{\vec{g}\in \Delta(S_i)} \dotp{\vec{g}}{\sum_{t=1}^{T} \vec{u}_i^t}-\frac{\mR(\vec{g})}{\eta}.  \]
These secondary variables are often called \emph{be the leader}
sequence as they can see one step in the future.

\begin{theorem}\label{thm:oftrl-bound} The regret of a player under optimistic FTRL and with respect to any $\vec{\mst}_i^*\in \Delta(S_i)$ is upper bounded by:
\begin{equation}
\sum_{t=1}^{T} \dotp{\vec{\mst}_i^*-\vec{\mst}_i^t}{\vec{u}_i^t} \leq \frac{R}{\eta} + \sum_{t=1}^T \|\vec{u}_i^t - \vec{M}_i^t\|_* \|\vec{\mst}_i^t - \vec{g}_i^t\| - \frac{1}{2\eta} \sum_{t=1}^{T} \left(\|\vec{\mst}_i^t-\vec{g}_i^t\|^2 + \|\vec{\mst}_i^t- \vec{g}_i^{t-1}\|^2\right)
\end{equation}
where $R =\sup_{\vec{f}} \mR(\vec{f})-\inf_{\vec{f}}\mR(\vec{f})$.
\end{theorem} 
\begin{proof}
First observe that:
\begin{equation}
\dotp{\vec{\mst}_i^* - \vec{\mst}_i^t}{\vec{u}_i^t} = \dotp{\vec{g}_i^{t} - \vec{\mst}_i^t}{\vec{u}_i^t-\vec{M}_i^t} + \dotp{\vec{g}_i^{t} - \vec{\mst}_i^t}{\vec{M}_i^t}+\dotp{\vec{\mst}_i^*-\vec{g}_i^{t}}{\vec{u}_i^t}
\end{equation}
Without loss of generality we will assume that $\inf_{\vec{f}}\mR(\vec{f})=0$. Since $ \dotp{\vec{g}_i^{t} - \vec{\mst}_i^t}{\vec{u}_i^t-\vec{M}_i^t}\leq \| \vec{g}_i^{t} - \vec{\mst}_i^t\|\|\vec{u}_i^t-\vec{M}_i^t\|_*$, it suffices to show that for any $\vec{\mst}_i^*\in \Delta(S_i)$:
\begin{equation}
\sum_{t=1}^{T} \left(\dotp{\vec{g}_i^{t} - \vec{\mst}_i^t}{\vec{M}_i^t}+\dotp{\vec{\mst}_i^*-\vec{g}_i^{t}}{\vec{u}_i^t}\right) \leq \frac{\mR(\vec{\mst}_i^*)}{\eta} - \frac{1}{2\eta} \sum_{t=1}^{T}\left( \|\vec{\mst}_i^t-\vec{g}_i^t\|^2 + \|\vec{\mst}_i^t - \vec{g}_i^{t-1}\|^2\right)
\end{equation}
For shorthand notation let: $I_T= \frac{1}{2\eta} \sum_{t=1}^{T}\left( \|\vec{\mst}_i^t-\vec{g}_i^t\|^2 + \|\vec{\mst}_i^t - \vec{g}_i^{t-1}\|^2\right)$.
By induction assume that for all $\vec{\mst}_i^*$:
\begin{align*}
\sum_{t=1}^{T-1} \left(\dotp{\vec{g}_i^{t} - \vec{\mst}_i^t}{\vec{M}_i^t}-\dotp{\vec{g}_i^{t}}{\vec{u}_i^t}\right) \leq~& -\sum_{t=1}^{T-1}\dotp{\vec{\mst}_i^*}{\vec{u}_i^t}+\frac{\mR(\vec{\mst}_i^*)}{\eta} - I_{T-1}\\
=~&-\dotp{\vec{\mst}_i^*}{\sum_{t=1}^{T-1}\vec{u}_i^t}+\frac{\mR(\vec{\mst}_i^*)}{\eta} - I_{T-1}
\end{align*}
Apply the above for $\vec{\mst}_i^*=\vec{g}_i^{T-1}$ and add $\dotp{\vec{g}_i^{T} - \vec{\mst}_i^T}{\vec{M}_i^T}-\dotp{\vec{g}_i^{T}}{\vec{u}_i^T}$ on both  sides: 
\begin{align*}
\sum_{t=1}^{T}\left( \dotp{\vec{g}_i^{t} - \vec{\mst}_i^t}{\vec{M}_i^t}-\dotp{\vec{g}_i^{t}}{\vec{u}_i^t}\right) \leq~& -\dotp{\vec{g}_i^{T-1}}{\sum_{t=1}^{T-1}\vec{u}_i^t}+\frac{\mR(\vec{g}_i^{T-1})}{\eta} - I_{T-1} + \dotp{\vec{g}_i^{T} - \vec{\mst}_i^T}{\vec{M}_i^T}-\dotp{\vec{g}_i^{T}}{\vec{u}_i^T}\\
\leq~&  -\dotp{\vec{\mst}_i^{T}}{\sum_{t=1}^{T-1}\vec{u}_i^t}+\frac{\mR(\vec{\mst}_i^{T})}{\eta} - I_{T-1} + \dotp{\vec{g}_i^{T} - \vec{\mst}_i^T}{\vec{M}_i^T}-\dotp{\vec{g}_i^{T}}{\vec{u}_i^T}\\
~& ~~~~-\frac{1}{2\eta}\|\vec{\mst}_i^T-\vec{g}_i^{T-1}\|^2 \\
=~& -\dotp{\vec{\mst}_i^{T}}{\sum_{t=1}^{T-1}\vec{u}_i^t+\vec{M}_i^T}+\frac{\mR(\vec{\mst}_i^{T})}{\eta} - I_{T-1} + \dotp{\vec{g}_i^{T}}{\vec{M}_i^T}-\dotp{\vec{g}_i^{T}}{\vec{u}_i^T}\\
~& ~~~~-\frac{1}{2\eta}\|\vec{\mst}_i^T-\vec{g}_i^{T-1}\|^2 \\
\leq~& -\dotp{\vec{g}_i^{T}}{\sum_{t=1}^{T-1}\vec{u}_i^t+\vec{M}_i^T}+\frac{\mR(\vec{g}_i^{T})}{\eta} - I_{T-1} + \dotp{\vec{g}_i^{T}}{\vec{M}_i^T}-\dotp{\vec{g}_i^{T}}{\vec{u}_i^T}\\
~& ~~~~-\frac{1}{2\eta}\|\vec{\mst}_i^T-\vec{g}_i^{T-1}\|^2-\frac{1}{2\eta}\|\vec{\mst}_i^T-\vec{g}_i^{T}\|^2  \\
=~& -\dotp{\vec{g}_i^{T}}{\sum_{t=1}^{T}\vec{u}_i^t}+\frac{\mR(\vec{g}_i^{T})}{\eta} - I_{T}\\
\leq~& -\dotp{\vec{q}_i^*}{\sum_{t=1}^{T}\vec{u}_i^t}+\frac{\mR(\vec{q}_i^*)}{\eta}- I_{T}
\end{align*}
The inequalities follow by the optimality of the corresponding variable that was changed and by the strong convexity of $\mR(\cdot)$. The final vector $\vec{q}_i^*$ is an arbitrary vector in $\Delta(S_i)$. The base case of $T=0$ follows trivially by $\mR(\vec{f})\geq 0$ for all $\vec{f}$. This concludes the inductive proof.
\end{proof}	

Thus optimistic FTRL achieves the exact same form of regret presented in Theorem \ref{thm:opt-omd-bound} for optimistic mirror descent. Hence, the equivalent versions of Theorem \ref{thm:opt-omd-better-bound} and Corollary \ref{cor:fast-omd} hold also for the optimistic FTRL algorithm. In fact we are able to show slightly stronger bounds for optimistic FTRL, based on the following lemmas.

\begin{lemma}[Stability]\label{lem:oftrl-stability}
For the optimistic FTRL algorithm:
\begin{align}
\|\vec{\mst}_i^t-\vec{g}_i^{t}\| \leq~& \eta \cdot \|\vec{M}_i^{t}-\vec{u}_i^{t}\|_*\\
\|\vec{g}_i^t-\vec{\mst}_i^{t+1}\| \leq~& \eta\cdot \|\vec{M}_i^{t+1}\|_*
\end{align}
\end{lemma}
\begin{proof}
Let $F_T(\vec{f}) = \dotp{\vec{f}}{\sum_{t=1}^{T-1} \vec{u}_i^t +\vec{M}_i^T}-\eta^{-1}\mR(\vec{f})$ and $G_T(\vec{f}) = \dotp{\vec{f}}{\sum_{t=1}^{T} \vec{u}_i^t} - \eta^{-1}\mR(\vec{f})$. Observe that: $F_T(\vec{f})-G_T(\vec{f}) = \dotp{\vec{f}}{\vec{M}_i^T-\vec{u}_i^T}$ and $F_{T+1}(\vec{f})-G_T(\vec{f}) = \dotp{\vec{f}}{\vec{M}_i^{T+1}}$. 

\paragraph{Part 1} By the optimality of $\vec{\mst}_i^T$ and $\vec{g}_i^T$ and the strong convexity of $\mR(\cdot)$:
\begin{align*}
F_T(\vec{\mst}_i^T) \geq~& F_T(\vec{g}_i^T) + \frac{1}{2\eta} \| \vec{\mst}_i^T - \vec{g}_i^T\|^2\\
G_T(\vec{g}_i^T) \geq~& G_T(\vec{\mst}_i^T) + \frac{1}{2\eta} \| \vec{\mst}_i^T - \vec{g}_i^T\|^2
\end{align*}
Adding both inequalities and using the previous observations:
\begin{align*}
\frac{1}{\eta} \| \vec{\mst}_i^T - \vec{g}_i^T\|^2\leq 
\dotp{\vec{\mst}_i^T - \vec{g}_i^T}{\vec{M}_i^T-\vec{u}_i^T}  \leq \|\vec{\mst}_i^T - \vec{g}_i^T\|\cdot \|\vec{M}_i^T-\vec{u}_i^T\|_*
\end{align*}
Dividing over by $\|\vec{\mst}_i^T - \vec{g}_i^T\|$ gives the first inequality of the lemma.

\paragraph{Part 2} By the optimality of $\vec{g}_i^T$ and $\vec{\mst}_i^{T+1}$ and strong convexity:
\begin{align*}
F_{T+1}(\vec{\mst}_i^{T+1}) \geq~& F_{T+1}(\vec{g}_i^T) + \frac{1}{2\eta} \| \vec{\mst}_i^{T+1} - \vec{g}_i^T\|^2\\
G_T(\vec{g}_i^T) \geq~& G_T(\vec{\mst}_i^{T+1}) + \frac{1}{2\eta} \| \vec{\mst}_i^{T+1} - \vec{g}_i^T\|^2
\end{align*}
Adding the inequalities:
\begin{align*}
\frac{1}{\eta} \| \vec{\mst}_i^{T+1} - \vec{g}_i^T\|^2\leq  \dotp{\vec{\mst}_i^{T+1}-\vec{g}_i^T}{\vec{M}_i^{T+1}} \leq \|\vec{\mst}_i^{T+1} - \vec{g}_i^T\|\cdot \|\vec{M}_i^{T+1}\|_*
\end{align*}
Dividing over by  $\| \vec{\mst}_i^{T+1} - \vec{g}_i^T\|$, yields second inequality of the lemma.
\end{proof}

Given Theorem \ref{thm:oftrl-bound} and Lemma
\ref{lem:oftrl-stability}, the proposition immediately follows since
\begin{equation*}
\sum_{t=1}^{T} \dotp{\vec{\mst}_i^*-\vec{\mst}_i^t}{\vec{u}_i^t} \leq \frac{R}{\eta} + \eta \sum_{t=1}^T \|\vec{u}_i^t - \vec{M}_i^t\|_*^2- \frac{1}{2\eta} \sum_{t=1}^{T} \left(\|\vec{\mst}_i^t-\vec{g}_i^t\|^2 + \|\vec{\mst}_i^t- \vec{g}_i^{t-1}\|^2\right).
\end{equation*}
Replacing $\vec{M}_i^t$ with $\vec{u}_i^{t-1}$ and using Inequality \eqref{eqn:squares}, yields the result.

\section{Proof of Proposition~\ref{prop:oftrl-Hstep}}

\begin{rtheorem}{Proposition}{\ref{prop:oftrl-Hstep}.}
  The OFTRL algorithm using stepsize $\eta$ and $\vec{M}_i^t
  = \sum_{\tau = t-H}^{t-1} \vec{u}_i^\tau/H$ satisfies
  the \myprop~property with constants $\alpha = R/\eta$, $\beta = \eta
  H^2$ and $\gamma = 1/(4\eta).$
\end{rtheorem}

The proposition is equivalent to the following lemma, which we will state and prove in this appendix.

\begin{lemma}\label{lem:finite-recency}
For the optimistic FTRL algorithm with $\vec{M}_i^t =\frac{1}{H}\sum_{\tau=t-H}^{t-1}\vec{u}_i^{\tau}$, the regret is upper bounded by:
\begin{equation}
\sum_{t=1}^{T} \dotp{\vec{\mst}_i^*-\vec{\mst}_i^t}{ \vec{u}_i^t}\leq \frac{R}{\eta} + \eta H^2 \sum_{t=1}^{T} \| \vec{u}_i^t - \vec{u}_i^{t-1}\|_*^2-\frac{1}{4\eta}\sum_{t=1}^T\|\vec{\mst}_i^t-\vec{\mst}_i^{t-1}\|^2
\end{equation}
where $R =\sup_{\vec{f}} \mR(\vec{f})-\inf_{\vec{f}}\mR(\vec{f})$.
Thus we get $\sum_i r_i(T)\leq \frac{n R}{\eta} = 2 n(n-1)H R$ for $\eta= \frac{1}{2 H (n-1)}$.
\end{lemma}
\begin{proof}
Similar to Proposition~\ref{prop:oftrl-onestep}, by Theorem
\ref{thm:oftrl-bound}, Lemma \ref{lem:oftrl-stability} and Inequality
\eqref{eqn:squares} we get:
\begin{align*}
\sum_{t=1}^{T} \dotp{\vec{\mst}_i^*-\vec{\mst}_i^t}{\vec{u}_i^t} \leq~& \frac{R}{\eta} + \eta \sum_{t=1}^T \|\vec{u}_i^t - \vec{M}_i^t\|_*^2-\frac{1}{4\eta}\sum_{t=1}^T\|\vec{\mst}_i^t-\vec{\mst}_i^{t-1}\|^2\\
=~&\frac{R}{\eta} + \eta \sum_{t=1}^T \left\|\vec{u}_i^t - \frac{1}{H}\sum_{\tau=t-H}^{t-1}\vec{u}_i^{\tau}\right\|_*^2-\frac{1}{4\eta}\sum_{t=1}^T\|\vec{\mst}_i^t-\vec{\mst}_i^{t-1}\|^2\\
=~&\frac{R}{\eta} + \eta \sum_{t=1}^T\left(\frac{1}{H}\sum_{\tau=t-H}^{t-1} \left\|\vec{u}_i^t - \vec{u}_i^{\tau}\right\|_*\right)^2-\frac{1}{4\eta}\sum_{t=1}^T\|\vec{\mst}_i^t-\vec{\mst}_i^{t-1}\|^2\\
\end{align*}
By triangle inequality:
\begin{align*}
\frac{1}{H}\sum_{\tau=t-H}^{t-1} \left\|\vec{u}_i^t - \vec{u}_i^{\tau}\right\|_* \leq~& \frac{1}{H}\sum_{\tau=t-H}^{t-1}\sum_{q=\tau}^{t-1} \left\|\vec{u}_i^{q+1} - \vec{u}_i^{q}\right\|_*\\
=~& \sum_{\tau=t-H}^{t-1}\frac{t-\tau}{H} \left\|\vec{u}_i^{\tau+1} - \vec{u}_i^{\tau}\right\|_*\leq \sum_{\tau=t-H}^{t-1} \left\|\vec{u}_i^{\tau+1} - \vec{u}_i^{\tau}\right\|_*
\end{align*}
By Cauchy-Schwarz:
\begin{align*}
\left(\sum_{\tau=t-H}^{t-1} \left\|\vec{u}_i^{\tau+1} - \vec{u}_i^{\tau}\right\|_*\right)^2\leq H \sum_{\tau=t-H}^{t-1} \left\|\vec{u}_i^{\tau+1} - \vec{u}_i^{\tau}\right\|_*^2
\end{align*}
Thus we can derive that:
\begin{align*}
\sum_{t=1}^{T} \dotp{\vec{\mst}_i^*-\vec{\mst}_i^t}{\vec{u}_i^t} \leq~& 
\frac{R}{\eta} + \eta H\sum_{t=1}^T \sum_{\tau=t-H}^{t-1} \left\|\vec{u}_i^{\tau+1} - \vec{u}_i^{\tau}\right\|_*^2-\frac{1}{4\eta}\sum_{t=1}^T\|\vec{\mst}_i^t-\vec{\mst}_i^{t-1}\|^2\\
\leq~& 
\frac{R}{\eta} + \eta H^2 \sum_{t=1}^T \left\|\vec{u}_i^{t} - \vec{u}_i^{t-1}\right\|_*^2-\frac{1}{4\eta}\sum_{t=1}^T\|\vec{\mst}_i^t-\vec{\mst}_i^{t-1}\|^2\\
\end{align*}
\end{proof}

\section{Proof of Proposition~\ref{prop:oftrl-discount}}

\begin{rtheorem}{Proposition}{\ref{prop:oftrl-discount}.}
  The OFTRL algorithm using stepsize $\eta$ and $\vec{M}_i^t
  = \frac{1}{\sum_{\tau = 0}^{t-1} \delta^{-\tau}} \sum_{\tau =
  0}^{t-1} \delta^{-\tau} \vec{u}_i^\tau$ satisfies
  the \myprop~property with constants $\alpha = R/\eta$, $\beta
  = \eta/(1- \delta)^3$ and $\gamma = 1/(8\eta).$
\end{rtheorem}

The proposition is equivalent to the following lemma which we will prove in this appendix.

\begin{lemma}\label{lem:discounted}
For the optimistic FTRL algorithm with $\vec{M}_i^t =\frac{1}{\sum_{\tau=0}^t \delta^{-\tau}}\sum_{\tau=0}^{t-1}\delta^{-\tau}\vec{u}_i^{\tau}$ for some discount rate $\delta\in (0,1)$, the regret is upper bounded by:
\begin{equation}
\sum_{t=1}^{T} \dotp{\vec{\mst}_i^*-\vec{\mst}_i^t}{ \vec{u}_i^t}\leq \frac{R}{\eta} + \frac{\eta}{(1-\delta)^3} \sum_{t=1}^{T} \| \vec{u}_i^t - \vec{u}_i^{t-1}\|_*^2-\frac{1}{8\eta}\sum_{t=1}^T\|\vec{\mst}_i^t-\vec{\mst}_i^{t-1}\|^2
\end{equation}
where $R =\sup_{\vec{f}} \mR(\vec{f})-\inf_{\vec{f}}\mR(\vec{f})$.
Thus we get $\sum_i r_i(T)\leq \frac{n R}{\eta} = 2 n(n-1)\frac{1}{(1-\delta)^{3/2}} R$ for $\eta= \frac{(1-\delta)^{3/2}}{2 (n-1)}$.
\end{lemma}
\begin{proof}
We show the theorem for the case of optimistic FTRL. The OMD case follows analogously. Similar to Lemma \ref{lem:finite-recency} the regret is upper bounded by:
\begin{align*}
\sum_{t=1}^{T} \dotp{\vec{\mst}_i^*-\vec{\mst}_i^t}{\vec{u}_i^t} \leq~& \frac{R}{\eta} + \eta \sum_{t=1}^T \|\vec{u}_i^t - \vec{M}_i^t\|_*^2-\frac{1}{4\eta}\sum_{t=1}^T\|\vec{\mst}_i^t-\vec{\mst}_i^{t-1}\|^2\\
=~&\frac{R}{\eta} + \eta \sum_{t=1}^T \left\|\vec{u}_i^t - \frac{1}{\sum_{\tau=0}^{t-1}\delta^{-\tau}}\sum_{\tau=0}^{t-1}\delta^{-\tau}\vec{u}_i^{\tau}\right\|_*^2-\frac{1}{4\eta}\sum_{t=1}^T\|\vec{\mst}_i^t-\vec{\mst}_i^{t-1}\|^2
\end{align*}
We will now show that:
\begin{align*}
 \sum_{t=1}^T \left\|\vec{u}_i^t - \frac{1}{\sum_{\tau=0}^{t-1}\delta^{-\tau}}\sum_{\tau=0}^{t-1}\delta^{-\tau}\vec{u}_i^{\tau}\right\|_*^2\leq \frac{1}{(1-\delta)^3} \sum_{t=1}^{T} \| \vec{u}_i^t - \vec{u}_i^{t-1}\|_*^2
\end{align*}
which will conclude the proof.

First observe by triangle inequality:
\begin{align*}
\left\|\vec{u}_i^t - \frac{1}{\sum_{\tau=0}^{t-1}\delta^{-\tau}}\sum_{\tau=0}^{t-1}\delta^{-\tau}\vec{u}_i^{\tau}\right\|_*=~& \frac{1}{\sum_{\tau=0}^{t-1}\delta^{-\tau}} \sum_{\tau=0}^{t-1} \delta^{-\tau}\| \vec{u}_i^t - \vec{u}_i^{\tau}\|_*\\
\leq~&\frac{1}{\sum_{\tau=0}^{t-1}\delta^{-\tau}} \sum_{\tau=0}^{t-1} \delta^{-\tau} \sum_{q=\tau}^{t-1} \left\|\vec{u}_i^{q+1} - \vec{u}_i^{q}\right\|_*\\
=~& \frac{1}{\sum_{\tau=0}^{t-1}\delta^{-\tau}} \sum_{q=0}^{t-1}\left\|\vec{u}_i^{q+1} - \vec{u}_i^{q}\right\|_*\sum_{\tau=0}^{q} \delta^{-\tau}\\
=~& \frac{1}{\sum_{\tau=0}^{t-1}\delta^{-\tau}} \sum_{q=0}^{t-1}\left\|\vec{u}_i^{q+1} - \vec{u}_i^{q}\right\|_* \delta^{-q}\frac{1-\delta^{q+1}}{1-\delta}\\
\leq~& \frac{1}{1-\delta}\frac{1}{\sum_{\tau=0}^{t-1}\delta^{-\tau}} \sum_{q=0}^{t-1}\delta^{-q}\left\|\vec{u}_i^{q+1} - \vec{u}_i^{q}\right\|_* 
\end{align*}
By Cauchy-Schwarz:
\begin{align*}
\left(\frac{1}{1-\delta}\frac{1}{\sum_{\tau=0}^{t-1}\delta^{-\tau}}\sum_{q=0}^{t-1}\delta^{-q}\left\|\vec{u}_i^{q+1} - \vec{u}_i^{q}\right\|_*\right)^2 =~& \frac{1}{(1-\delta)^2}\frac{1}{\left(\sum_{\tau=0}^{t-1}\delta^{-\tau}\right)^2}\left(\sum_{q=0}^{t-1} \delta^{-q/2}\cdot \delta^{-q/2}\left\|\vec{u}_i^{q+1} - \vec{u}_i^{q}\right\|_* \right)^2\\
\leq~& \frac{1}{(1-\delta)^2}\frac{1}{\left(\sum_{\tau=0}^{t-1}\delta^{-\tau}\right)^2}\sum_{q=0}^{t-1} \delta^{-q}\cdot \sum_{q=0}^{t-1} \delta^{-q} \left\|\vec{u}_i^{q+1} - \vec{u}_i^{q}\right\|_*^2\\
=~&\frac{1}{(1-\delta)^2}\frac{1}{\sum_{\tau=0}^{t-1}\delta^{-\tau}} \sum_{q=0}^{t-1} \delta^{-q} \left\|\vec{u}_i^{q+1} - \vec{u}_i^{q}\right\|_*^2\\
=~&\frac{1}{(1-\delta)^2}\frac{1}{\sum_{\tau=0}^{t-1}\delta^{t-\tau}} \sum_{q=0}^{t-1} \delta^{t-q} \left\|\vec{u}_i^{q+1} - \vec{u}_i^{q}\right\|_*^2\\
\leq~&\frac{1}{\delta(1-\delta)^2}\sum_{q=0}^{t-1} \delta^{t-q} \left\|\vec{u}_i^{q+1} - \vec{u}_i^{q}\right\|_*^2\\
\end{align*}
Combining we get:
\begin{align*}
\left\|\vec{u}_i^t - \frac{1}{\sum_{\tau=0}^{t-1}\delta^{-\tau}}\sum_{\tau=0}^{t-1}\delta^{-\tau}\vec{u}_i^{\tau}\right\|_*^2 \leq \frac{1}{\delta (1-\delta)^2}\sum_{q=0}^{t-1} \delta^{t-q} \left\|\vec{u}_i^{q+1} - \vec{u}_i^{q}\right\|_*^2
\end{align*}
Summing over all $t$ and re-arranging we get:
\begin{align*}
\sum_{t=1}^T \left\|\vec{u}_i^t - \frac{1}{\sum_{\tau=0}^{t-1}\delta^{-\tau}}\sum_{\tau=0}^{t-1}\delta^{-\tau}\vec{u}_i^{\tau}\right\|_*^2\leq~&
\frac{1}{\delta (1-\delta)^2} \sum_{t=1}^{T} \sum_{q=0}^{t-1} \delta^{t-q} \left\|\vec{u}_i^{q+1} - \vec{u}_i^{q}\right\|_*^2\\
=~& 
\frac{1}{\delta(1-\delta)^2} \sum_{q=0}^{T-1}\delta^{-q} \left\|\vec{u}_i^{q+1} - \vec{u}_i^{q}\right\|_*^2  \sum_{t=q+1}^{T} \delta^{t}\\
=~& 
\frac{1}{\delta(1-\delta)^2} \sum_{q=0}^{T-1}\delta^{-q} \left\|\vec{u}_i^{q+1} - \vec{u}_i^{q}\right\|_*^2  \frac{\delta(\delta^{q}-\delta^{T})}{1-\delta}\\
=~& 
\frac{1}{(1-\delta)^3} \sum_{q=0}^{T-1}\left\|\vec{u}_i^{q+1} - \vec{u}_i^{q}\right\|_*^2  (1-\delta^{T-q})\\
\leq~& 
\frac{1}{(1-\delta)^3} \sum_{q=0}^{T-1}\left\|\vec{u}_i^{q+1} - \vec{u}_i^{q}\right\|_*^2
\end{align*}

\end{proof}

\section{Proof of Theorem~\ref{thm:adaptive}}

\begin{rtheorem}{Theorem}{\ref{thm:adaptive}.}
Algorithm $\A'$ achieves regret at most the minimum of the following
two terms:
\begin{small}
\begin{align*}
\sum_{t=1}^{T} \dotp{\vec{\mst}_i^*-\vec{\mst}_i^t}{\vec{u}_i^t} \leq& \log(T)\left(2+\frac{\alpha}{\eta_*} + (2+\eta_*\cdot \beta)\sum_{t=1}^T \|\vec{u}_i^t - \vec{u}_i^{t-1}\|_*^2\right) - \frac{\gamma}{\eta_*}\sum_{t=1}^{T} \|\vec{\mst}_i^t-\vec{\mst}_i^{t-1}\|^2;\\
\sum_{t=1}^{T} \dotp{\vec{\mst}_i^*-\vec{\mst}_i^t}{\vec{u}_i^t} \leq& \log(T)\left(1+\frac{\alpha}{\eta_*} + (1+\alpha\cdot\beta)\cdot\sqrt{2\sum_{t=1}^T \|\vec{u}_i^t - \vec{u}_i^{t-1}\|_* ^2}\right)
\end{align*}
\end{small}
\end{rtheorem}
\begin{proof}
We break the proof in the two corresponding parts.

\paragraph{First part.} Consider a round $r$ and let $T_r$ be its final iteration. Also let $I_r =\sum_{t=1}^{T_r} \|\vec{u}_i^t - \vec{u}_i^{t-1}\|_*^2$.  First observe that by the definition of $B_r$:
\begin{equation}\label{eqn:bound-squash}
\frac{1}{2}I_r\leq B_r \leq 2\cdot I_r +1
\end{equation}
By the definition of $\eta$, we know that 
\begin{equation}\label{eqn:eta-bound}
\frac{1}{\eta_*}\leq \frac{1}{\eta}< \frac{1}{\eta_*}+\frac{\sqrt{B_r}}{\alpha}.
\end{equation}
By the regret guarantee of algorithm $\A(\eta_r)$, we have that:
\begin{align*}
\sum_{t=T_{r-1}+1}^{T_r} \dotp{\vec{\mst}_i^*-\vec{\mst}_i^t}{\vec{u}_i^t} \leq~&\frac{\alpha}{\eta}+\eta\cdot \beta\sum_{t=T_{r-1}+1}^{T_r} \|\vec{u}_i^t - \vec{u}_i^{t-1}\|_*^2  - \frac{\gamma}{\eta}\sum_{t=T_{r-1}+1}^{T_r} \|\vec{\mst}_i^t-\vec{\mst}_i^{t-1}\|^2\\
\leq~& \frac{\alpha}{\eta_*} + \sqrt{B_r}+\eta_*\cdot \beta\sum_{t=T_{r-1}+1}^{T_r} \|\vec{u}_i^t - \vec{u}_i^{t-1}\|_*^2  - \frac{\gamma}{\eta_*}\sum_{t=T_{r-1}+1}^{T_r} \|\vec{\mst}_i^t-\vec{\mst}_i^{t-1}\|^2\\
\leq~&\frac{\alpha}{\eta_*} + \sqrt{B_r}+\eta_*\cdot \beta\sum_{t=1}^{T} \|\vec{u}_i^t - \vec{u}_i^{t-1}\|_*^2  - \frac{\gamma}{\eta_*}\sum_{t=T_{r-1}+1}^{T_r} \|\vec{\mst}_i^t-\vec{\mst}_i^{t-1}\|^2
\end{align*}
Since $\sqrt{B_r}\leq B_r+1\leq 2\cdot I_r+2$:
\begin{align*}
\sum_{t=T_{r-1}+1}^{T_r} \dotp{\vec{\mst}_i^*-\vec{\mst}_i^t}{\vec{u}_i^t}
\leq~&\frac{\alpha}{\eta_*} + 2+ (2 +\eta_*\cdot \beta)\sum_{t=1}^{T} \|\vec{u}_i^t - \vec{u}_i^{t-1}\|_*^2  - \frac{\gamma}{\eta_*}\sum_{t=T_{r-1}+1}^{T_r} \|\vec{\mst}_i^t-\vec{\mst}_i^{t-1}\|^2
\end{align*}
Since at each round we are doubling the bound $B_r$ and since $\sum_{t=1}^{T}\|\vec{u}_i^t - \vec{u}_i^{t-1}\|_*^2 \leq T$, there are at most $\log(T)$ rounds. Summing up the above inequality for each of the at most $\log(T)$ rounds, yields the claimed bound in Equation \eqref{eqn:nice}.

\paragraph{Second part.} Again consider any round $r$. By Equations \eqref{eqn:bound-squash}, \eqref{eqn:eta-bound}, the fact that $\eta\leq \frac{\alpha}{\sqrt{B_r}}\leq \frac{\alpha\sqrt{2}}{\sqrt{I_r}}$ and by the regret of algorithm $\A(\eta_r)$:
\begin{align*}
\sum_{t=T_{r-1}+1}^{T_r} \dotp{\vec{\mst}_i^*-\vec{\mst}_i^t}{\vec{u}_i^t}\leq~&  \frac{\alpha}{\eta_*} + \sqrt{B_r}+\eta\cdot \beta\sum_{t=T_{r-1}+1}^{T_r} \|\vec{u}_i^t- \vec{u}_i^{t-1}\|_*^2  \\
\leq~&  \frac{\alpha}{\eta_*} + \sqrt{B_r}+\eta\cdot \beta \cdot I_r\\
\leq~&  \frac{\alpha}{\eta_*} + \sqrt{B_r}+\alpha\cdot \beta \cdot\sqrt{2 I_r}\\
\leq~&  \frac{\alpha}{\eta_*} + \sqrt{2 I_r+1}+\alpha\cdot \beta \cdot\sqrt{2 I_r}\\
\leq~&  \frac{\alpha}{\eta_*} + 1+\sqrt{2 I_r}+\alpha\cdot \beta \cdot\sqrt{2 I_r}\\
\leq~& \frac{\alpha}{\eta_*} +1+(1+\alpha\cdot \beta)\sqrt{2\sum_{t=1}^T \|\vec{u}_i^t - \vec{u}_i^{t-1}\|_*^2}
\end{align*}
Again since the number of rounds is at most $\log(T)$, by summing up the above bound for each round $r$, we get the second part of the theorem.
\end{proof}

\section{Proof of Corollary~\ref{cor:setting-individual-const}}

\begin{rtheorem}{Corollary}{\ref{cor:setting-individual-const}.}
If $\A$ satisfies the $\myprop(\knob)$~property, and also
$\|\vec{\mst}_i^t-\vec{\mst}_i^{t-1}\| \leq \stable\knob$, then $\A'$
with $\eta_*=T^{-1/4}$ achieves regret $\tilde{O}(T^{1/4})$ if played
against itself, and $\tilde{O}(\sqrt{T})$ against any opponent.
\end{rtheorem}
\begin{proof}
Observe that at any round of $\A'$, algorithm $\A$ is run with
$\eta_r\leq \eta_*$. Thus by the property of algorithm $\A$, we have
that at every iteration:
$\|\vec{\mst}_i^t-\vec{\mst}_i^{t-1}\| \leq \stable \eta_* = \stable
T^{-1/4}$. If all players use algorithm $\A'$, then by similar
reasoning as in Theorem \ref{thm:sufficient} we know that:
\begin{equation*}
 \| \vec{u}_i^t - \vec{u}_i^{t-1}\|_*^2 \leq (n-1) \sum_{j\neq
 i} \|\vec{\mst}_j^t- \vec{\mst}_j^{t-1}\|^2\leq
 (n-1)^2 \gamma^2 \eta_*^2 = (n-1)^2 \stable^2
 T^{-1/2} \end{equation*} Hence, by Equation \ref{eqn:adversarial},
 the regret of each player is bounded by:
\begin{align*}
\sum_{t=1}^{T} \dotp{\vec{\mst}_i^*-\vec{\mst}_i^t}{\vec{u}_i^t} \leq& \log(T)\left(\frac{\alpha}{\eta_*} + (1+\alpha\cdot\beta)\cdot\sqrt{\sum_{t=1}^T \|\vec{u}_i^t - \vec{u}_i^{t-1}\|_* ^2}\right)\\
\leq& \log(T)\left(\alpha T^{1/4} + (1+\alpha\cdot\beta)\cdot\sqrt{T\cdot  (n-1)^2 \stable^2 T^{-1/2}}\right)\\
=& \log(T)\left(\alpha T^{1/4} +
(1+\alpha\cdot\beta)\cdot(n-1) \stable T^{1/4}\right)
= \tilde{O}(T^{1/4})
\end{align*} 
\end{proof}


\section{Fast convergence via a first order regret bound for cost-minimization}\label{sec:first-order}

In this section, we show how a different regret bound can also lead to a fast convergence rate for a smooth game. For some technical reasons we consider cost instead of utility throughout this section. We use $c_i: S_1\times \ldots \times S_n \rightarrow [0,1]$ to denote the cost function, and similarly to previous sections $C(\vec{s}) = \sum_{i\in N} c_i(\vec{s}), 
C(\vec{\mst}) = \E_{\vec{\st}\sim \vec{\mst}}[C(\vec{\st})], 
\opt' = \min_{\vec{\st}\in S_1\times \ldots \times S_n} C(\vec{s})$.
A game is $(\lambda,\mu)$-smooth if there exists a strategy profile $\vec{\st}^*$, such that for any strategy profile $\vec{\st}$:
\begin{equation}
\sum_{i\in N} c_i(\st_i^*,\vec{\st}_{-i}) \leq \lambda \opt' + \mu C(\vec{\st}).
\end{equation}
Now suppose each player $i$ uses a no-regret algorithm to produce $\vec{\mst}_i^t$ on each round and receives cost $c_{i,s}^t = \E_{\vec{s}_{-i}\sim \vec{\mst}_{-i}^t}[c_i(s,\vec{s}_{-i})]$ for each strategy $s \in S_i$. Moreover, for any fixed strategy $s$, the no-regret algorithm ensures
\begin{equation}\label{eqn:first_order_reg}
\sum_{t=1}^T \dotp{\vec{w}_i^t}{\vec{c}_i^t} - \sum_{t=1}^T c^t_{i, s} \leq 
A_1 \sqrt{\log d\left(\sum_{t=1}^T c^t_{i, s} \right)} + A_2\log d
\end{equation}
for some absolute constants $A_1$ and $A_2$. Note that this form of first order bound can be achieved by a variety of algorithms such as Hedge with appropriate learning rate tuning. Under this setup, we prove the following:

\begin{theorem}
If a game is $(\lambda,\mu)$-smooth and each player uses a no-regret algorithm with a regret satisfying Eq.~\eqref{eqn:first_order_reg}, then we have
\[ \frac{1}{T}\sum_{t=1}^T C(\vec{\mst}^t) \leq \frac{\lambda(1+\mu)}{\mu(1-\mu)}\opt' + \frac{An\log d}{T}\]
where $A = \frac{A_1^2\mu}{(1-\mu)^2}+\frac{2A_2}{1-\mu}$.
\end{theorem}

\begin{proof}
Using the regret bound and Cauchy-Schwarz inequality, we have
\begin{align}
\sum_{t=1}^T C(\vec{\mst}^t) &= \sum_{t=1}^T\sum_{i\in N} \dotp{\vec{w}_i^t}{\vec{c}_i^t} \notag\\
&\leq \sum_{t=1}^T\sum_{i\in N} c^t_{i, s_i^*} + A_1\sqrt{\log d} \sum_{i\in N} \sqrt{\sum_{t=1}^T c^t_{i, s_i^*}} + A_2n\log d \notag\\
&\leq \sum_{t=1}^T\sum_{i\in N} c^t_{i, s_i^*} + A_1\sqrt{n\log d} \sqrt{\sum_{T=1}^T\sum_{i\in N} c^t_{i, s_i^*}} + A_2n\log d. \label{eqn:ineq1}
\end{align}
By the smoothness assumption, we have 
\[ \sum_{i\in N} c^t_{i, s_i^*} =  \E_{\vec{s}\sim \vec{\mst}^t}\left[\sum_{i\in N}c_i(\st_i^*,\vec{\st}_{-i})\right]  \leq \lambda\opt' + \mu\E_{\vec{s}\sim \vec{\mst}^t}[C(\vec{s})] = \lambda\opt' + \mu C(\vec{\mst}^t), \]
and therefore $\sum_{t=1}^T\sum_{i\in N} c^t_{i, s_i^*} \leq x^2$ where we define $x = \sqrt{\lambda T\opt' + \mu\sum_{t=1}^TC(\vec{\mst}^t)}$. 
Now applying this bound in Eq.~\eqref{eqn:ineq1}, we continue with
\[ \frac{1}{\mu}\left(x^2 - \lambda T\opt' \right) \leq x^2 + (A_1\sqrt{n\log d}) x + A_2n\log d.\]
Rearranging gives a quadratic inequality $ax^2+bx+c \leq 0$ with 
\[a = \frac{1-\mu}{\mu}, \quad b = -A_1\sqrt{n\ln d}, \quad c = -\frac{\lambda}{\mu}T\opt' - A_2n\log d, \]
and solving for $x$ gives
\[x \leq \frac{\mu}{2(1-\mu)}(-b+\sqrt{b^2-4ac}) \leq \frac{\mu}{1-\mu}\sqrt{b^2 - 2ac}. \]
Finally solving for $\sum_{t=1}^TC(\vec{\mst}^t)$ (hidden in the definition of $x$) gives the bound stated in the theorem.
\end{proof}

Note that the price of total anarchy is larger than the one achieved by previous analysis by a multiplicative factor of $1+\frac{1}{\mu}$, but the convergence rate is much faster ($n$ times faster compared to optimistic mirror descent or optimistic FTRL).


\section{Extension to continuous strategy space games}

In this section we extend our results to continuous strategy space games such as for instance "splittable selfish routing games" (see e.g. \cite{Roughgarden2011}). These are games where the price of anarchy has been well studied and quite well motivated from internet routing. In these games we consider the dynamics where the players simply observe the past play of their opponents and not the expected past play. We consider dynamics where players don't use mixed strategies, but are simply doing online convex optimization algorithms on their continuous strategy spaces. Such learning on continuous games has also been studied in more restrictive settings in \cite{Even-dar2009}.

In this setting we will consider the following setting: each player $i$ has a strategy space $S_i$ which is a closed convex set in $\R^d$. In this setting we will denote with $\vec{w}_i\in S_i$
a strategy of a player\footnote{We will use $\vec{w}_i$ instead of $s_i$ for a pure strategy, since pure strategies of the continuous game will be sort of treated equivalently to mixed strategies in the discrete game we described in Section \ref{sec:prelims}}. Given a profile of strategies $\vec{w}=(\vec{w}_1,\ldots,\vec{w}_n)$, each player incurs a cost $c_i(\vec{w})$ (equivalently a utility function $u_i(\vec{w})$.

We make the following two assumptions on the costs:
\begin{enumerate}
\item (Convex in player strategy) For each player $i$ and for each profile of opponent strategies $\vec{w}_{-i}$, the function $c_i(\cdot, \vec{w}_{-i})$ is convex in $\vec{w}_i$.
\item (Lipschitz gradient) For each player $i$, the function $\delta_i(\vec{w}) = \nabla_{i}c_i(\vec{w})$,\footnote{With $\nabla_{i} c_i(\vec{w})$ we denote the gradient of the function with respect the strategy of player $i$ and fixing the strategy of other players. Equivalently for each fixed $\vec{w}_{-i}$ it is the gradient of the function $c_i(\cdot, \vec{w}_{-i})$.} is $L$-Lipschitz continuous with respect to the $\|\cdot\|_1$ norm and if $\vec{w}_{-i}\in \R^{(n-1)d}$ is viewed as a vector in the $(n-1)\cdot d$ dimensional space, i.e.:
\begin{equation}
\|\delta_i(\vec{w}) - \delta_i(\vec{y})\|_{*} \leq L\cdot \sum_{j} \|\vec{w}_j-\vec{y}_j\|
\end{equation}
\end{enumerate}

Observe that a sufficient condition for Property (2) is that the function $\delta_i(\vec{w})$ is coordinate-wise $L$-lipschitz with respect to the $\|\cdot\|$ norm.
\begin{lemma}
If for any $j$:
\begin{equation}
\|\delta_i(\vec{w})-\delta_i(\vec{y}_j,\vec{w}_{-j})\|_{*}\leq L\|\vec{w}_j-\vec{y}_j\|
\end{equation}
then $\delta_i(\cdot)$ satisfies Property (2).
\end{lemma}
\begin{proof}
For any two vectors $\vec{w}$ and $\vec{y}$, think of switching from the one to the other by switching sequentially each player from his strategy $\vec{w}_i$ to $\vec{y}_i$, keeping the remaining players fixed and in some pre-fixed player order. The difference
$\|\delta_i(\vec{w})-\delta_i(\vec{y})\|_*$ is upper bounded by the sum of the differences of these sequential switches. The difference of each such unilateral switch for each player $j$ is turn upper bounded by $\|\vec{w}_j-\vec{y}_j\|$, by the property assumed in the Lemma. The lemma then follows.
\end{proof}

\begin{example}(Connection to discrete game). We can  view the discrete action games as a special case of the latter setting, by re-naming mixed strategies in the discrete game to pure strategies in the continuous space game. Under this mapping, the continuous strategy space is the simplex in $\R^d$, where $d$ is the number of pure strategies of the discrete game. Moreover the costs $c_i(\vec{w})$ (equiv. utilities) are multi-linear, i.e. $c_i(\vec{w}) = \sum_{s} C_i(s)\prod_{j} w_{j,s}$. Obviously, these multi-linear costs satisfy assumption $1$, i.e. they are convex (in fact linear) in a players strategy. 

The second assumption is also satisfied, albeit with a slightly more involved proof, which appears in the proof of Theorem \ref{thm:sufficient}. Basically, observe that 
\begin{equation}
\delta_{i,s_i}(\vec{w})=\sum_{s_{-i}}C_i(s_i,s_{-i})\prod_{j\neq i}w_{j,s_j}
\end{equation}
Assuming $C_i(s)\leq 1$:
\begin{equation}
|\delta_{i,s_i}(\vec{w})- \delta_{i,s_i}(\vec{y})| \leq \sum_{\vec{s}_{-i}} \left| \prod_{j\neq i} \mst_{j,s_j} - \prod_{j\neq i}\vec{y}_{j,s_j}\right|\leq \sum_{j\neq i} \|\vec{w}_j-\vec{y}_j\|_1
\end{equation}
Where the last inequality holds by the properties of total variation distance.
\end{example}

\begin{example}(Splittable congestion games). In this game each player $i$ has an amount of flow $f_i\leq B$ he wants to route from a source $s_i$ to a sink $t_i$ in an undirected graph $G=(V,E)$. Each edge $e\in E$ is associated with a latency function $\ell_e(f_e)$ which maps an amount of flow $f_e$ passing through the edge to a latency. We will assume that latency functions are convex, increasing and twice differentiable. We will also assume that both $\ell_e(\cdot)$ and $\ell_e'(\cdot)$ are $K$-lipschitz functions of the flow. We will denote with $\mP_i$ the set of $(s_i,t_i)$ paths in the graph. Then the set of feasible strategies for each player is all possible ways of splitting his flow $f_i$ onto these pats $\mP_i$. Denote with $w_p$ the amount of flow a player routes on path $p\in \mP_i$, then the strategy space is:
\begin{equation}
S_i = \left\{\vec{w}_i\in \R^{|\mP_i|}_+: \sum_{p\in \mP_i}w_{i,p}=f_p\right\}
\end{equation}
The latter is obviously a closed convex set in $\R^{|\mP_i|}$. 

For an edge $e$, let $f_{i,e}(\vec{w}_i)=\sum_{p\in \mP_i: e\in p} w_{i,p}$ to be the flow on edge $e$ caused by player $i$ and with $f_e(\vec{w}) = \sum_{i}f_{i,e}(\vec{w}_i)$ to be the total flow on the edge $e$. Then the cost of a player is:
\begin{equation}
c_i(\vec{w}) = \sum_{e} f_{i,e}(\vec{w}_i)\cdot \ell(f_e(\vec{w}))
\end{equation}

First observe that the functions $c_i(\cdot)$ are convex with respect to a player's strategy $\vec{w}_i$. This follows since the cost is linear across edges, thus we need to show convexity locally at each edge. The latency function on an edge is a convex function of the total flow, hence also $x\ell_e(x+b)$ is also a convex function of $x$. Now observe that the cost from each edge is of the form $f_{i,e}(\vec{w}_{i}) \ell_e(f_{i,e}(\vec{w}_{i})+b)$ which is convex with respect to $f_{i,e}(\vec{w}_{i})$. In turn, $f_{i,e}(\cdot)$ is a linear function of $\vec{w}_i$. Thus whole cost function is convex in $\vec{w}_i$.

Last we need to show that the second condition on the cost functions is satisfied for some lipschitzness factor $L$. This will be a consequence of the $K$-lipschitzness of the latency functions. Denote with $\ell_e^i(\vec{w}) =\ell_e(f_e(\vec{w}))+f_{i,e}(\vec{w}_i) \cdot \ell_e'(f_e(\vec{w}))$. Then, observe that:
\begin{equation}
\delta_{i,p}(\vec{w}) = \sum_{e\in p} \left(\ell_e(f_e(\vec{w}))+f_{i,e}(\vec{w}_i) \cdot \ell_e'(f_e(\vec{w}))\right)= \sum_{e\in p}\ell_e^i(\vec{w})
\end{equation}
  Since both $\ell_e(\cdot)$ and $\ell_e'(\cdot)$ are $K$-lipschitz and $f_{i,e}(\vec{w}_i)\leq B$, we have that:
\begin{align*}
|\delta_{i,p}(\vec{w})- \delta_{i,p}(\vec{y})| \leq~& \sum_{e\in p} |\ell_e^i(\vec{w})-\ell_e^i(\vec{y})|\\
\leq~& \sum_{e\in p} |\ell_e(f_e(\vec{w}))-\ell_e(f_e(\vec{y}))|+B\sum_{e\in p} |\ell_e'(f_e(\vec{w}))-\ell_e'(f_e(\vec{y}))|\\
\leq~& K(1+B) \sum_{e\in p} |f_e(\vec{w}) - f_e(\vec{y})| \leq K(1+B) \sum_{e\in p}  \sum_{j\in [n]} |f_{j,e}(\vec{w}_j)-f_{j,e}(\vec{y}_j)|\\
\leq~& K(1+B) \sum_{e\in p}  \sum_{j\in [n]} \sum_{q\in \mP_j:e\in q}|w_{j,q}-y_{j,q}|\\
=~& K(1+B)  \sum_{j\in [n]} \sum_{q\in \mP_j} \sum_{e\in p\cap q}|w_{j,q}-y_{j,q}| \leq K(1+B)m \sum_{j\in [n]} \sum_{q\in \mP_j} |w_{j,q}-y_{j,q}|\\
\leq~& K(1+B)m \sum_{j\in [n]} \|\vec{w}_j-\vec{y}_j\|_1
\end{align*}
Thus we get that the second condition is satisfied with $L=2Km$.
\end{example}

For these games we will assume that the players are performing some form of regularized learning using the gradients of their utilities as proxies. For fast convergence we would require that the algorithms they use satisfy the following property, which is a generalization of Theorem \ref{thm:sufficient}.
\begin{theorem}\label{thm:sufficient-general}
Consider a repeated continuous strategy space game where the cost functions satisfy properties $1,2$. Suppose that the algorithm of each player $i$ satisfies the property that for any $\vec{\mst}_i^*\in S_i$
\begin{equation}
\sum_{t=1}^{T} c_i(\vec{w}^t) - c_i(\vec{w}_i^*,\vec{w}_{-i}^t) \leq \alpha +\beta \sum_{t=1}^{T} \| \delta_i(\vec{w}^t) - \delta_i(\vec{w}^{t-1})\|_* ^2- \gamma\sum_{t=1}^{T} \|\vec{w}_{i}^{t} - \vec{w}_{i}^{t-1}\|^2
\end{equation}
for some $\alpha>0$ and $0<\beta\leq \frac{\gamma}{L^2\cdot n^2}$ and with $\|\cdot \|$ we denote the $\|\cdot\|_1$ norm. Then:
\begin{equation}
\sum_{i\in N} r_i(T) \leq n\cdot \alpha = O(1)
\end{equation} 
\end{theorem}
\begin{proof}
By property $2$, we have that: 
\begin{align*}
\sum_{t=1}^{T} \| \delta_i(\vec{w}^t) - \delta_i(\vec{w}^{t-1})\|_* ^2 \leq L^2\sum_{t=1}^{T}\left(\sum_{j\in [n]} \|\vec{w}_j^t-\vec{w}_j^{t-1}\|\right)^2\leq
L^2 n\sum_{t=1}^{T}\sum_{j\in [n]}  \|\vec{w}_j^t-\vec{w}_j^{t-1}\|^2
\end{align*}
By summing up the regret inequality for each player and using the above bound we get: 
\begin{equation}
\sum_{i\in N} r_i(T) \leq n\cdot \alpha + \beta L^2 n^2 \sum_{t=1}^{T}\sum_{j\in [n]}  \|\vec{w}_j^t-\vec{w}_j^{t-1}\|^2- \gamma\sum_{t=1}^{T}\sum_{i\in [n]} \|\vec{w}_{i}^{t} - \vec{w}_{i}^{t-1}\|^2
\end{equation} 
If $\beta L^2 n^2 \leq \gamma$, the theorem follows.
\end{proof}

All the algorithms that we described in the previous sections can be adapted to satisfy the bound required by Theorem \ref{thm:sufficient-general}, by simply using the gradient of the cost as a proxy of the cost instead of the actual cost. This follows by standard arguments. Hence if players follow for instance the following adaptation of the regularized leader algorithm:
\begin{equation}
\vec{\mst}_i^T = \argmax_{\vec{\mst} \in S_i}  \dotp{\vec{\mst}}{\sum_{t=1}^{T-1} \delta_i(\vec{w}^t) + \delta_i(\vec{w}^{T-1})}-\frac{\mR(\vec{\mst})}{\eta}
\end{equation}
then by Proposition \ref{prop:oftrl-onestep} we get that their regret satisfies the conditions of Theorem \ref{thm:sufficient-general} for $a=\frac{R}{\eta}$, $\beta=\eta$ and $\gamma=\frac{1}{4\eta}$, where $R=\argmax_{\vec{w}_i\in S_i} R(\vec{w}_i)$. We need that $\eta L^2 n^2\leq \frac{1}{4\eta}$ or equivalently $\eta\leq \frac{1}{2 Ln}$. Thus for $\eta=\frac{1}{2Ln}$, if all players are using the latter algorithm we get regret of at most $n\cdot \frac{R}{\eta}=2 L n^2 R$

\begin{example}(Splittable congestion games).
Consider the case of congestion games with splittable flow, where all the latencies and their derivatives are $K$-Lipschitz and the flow of each player is at most $B$. In that setting, suppose that we use the entropic regularizer. Then for each player $i$, $R\leq B\cdot \log(|\mP_i|)$. The number of possible $(s,t)$ paths is at most $2^m$, which yields $R\leq B\cdot m$. 
Hence, by using the linearized follow the regularized leader, we get that the total regret is at most $2Ln^2R \leq 2 K (B+1)B m^2 n ^2$.
\end{example}

\section{$\Omega(\sqrt{T})$ Lower Bounds on Regret for other Dynamics}
We consider a two-player zero-sum game which can be described by a utility matrix $A$. Assume the row player uses MWU with a fixed learning rate $\eta$, and the column player plays the best response, that is, a pure strategy that minimizes the row player's expected utility for the current round. Then the following theorem states that no matter how $\eta$ is set, there is always a game $A$ such that the regret of the row player is at least $\Omega(\sqrt{T})$.

\begin{theorem}
In the setting described above, let $r(T)$ and $r'(T)$ be the regret of the row player for the game $A = \begin{pmatrix} 1 & 0 \\ 0 & 1 \end{pmatrix}$ and $A' = \begin{pmatrix} 1 \\ 0 \end{pmatrix}$ respectively after $T$ rounds. Then $\max\{r(T), r'(T)\} \geq \Omega(\sqrt{T})$.
\end{theorem}

\begin{proof}
For game $A$, according to the setup, one can verify that the row player will play a uniform distribution and receive utility $\frac{1}{2}$ on round $t$ where $t$ is odd, and for the next round $t+1$, the row player will put slightly more weights on one row and the column player will pick the column that has $0$ utility for that row. Specifically, the expected utility of the row player is $\frac{e^{\eta(t-1)/2}}{e^{\eta(t-1)/2}+e^{\eta(t+1)/2}} = \frac{1}{1+e^{\eta}}$. Therefore, the regret is (assuming $T$ is even for simplicity)
\[ r(T) = \frac{T}{2} - \frac{T}{2}\left(\frac{1}{2}+\frac{1}{1+e^{\eta}}\right) 
=  \frac{T}{2}\cdot\frac{e^\eta-1}{e^\eta+1}. \]
For game $A'$, the expected utility of the row player on round $t$ is $\frac{e^{\eta(t-1)}}{e^{\eta(t-1)}+1}$, and thus the regret is 
\[ r'(T) = T - \sum_{t=1}^T \frac{e^{\eta(t-1)}}{e^{\eta(t-1)}+1} = \sum_{t=1}^T \frac{1}{e^{\eta(t-1)}+1} 
\geq \sum_{t=1}^T \frac{1}{2e^{\eta(t-1)}} =  \frac{1-e^{-T\eta}}{2(1-e^{-\eta})}. \]
Now if $\eta \geq 1$, then $r(T) \geq \frac{T}{2}\cdot\frac{e-1}{e+1} = \Omega(T)$. If $\eta \leq \frac{1}{T}$, then $r'(T) \geq \frac{1-e^{-1}}{2(1-e^{-\frac{1}{T}})} \geq \frac{T(1-e^{-1})}{2} = \Omega(T)$. Finally when $\frac{1}{T} \leq \eta \leq 1$, we have
\[ r(T) + r'(T) \geq  \frac{T}{2}\cdot\frac{e^\eta-1}{e+1} + \frac{1-e^{-1}}{2(1-e^{-\eta})} 
\geq \frac{T}{2}\cdot\frac{e^\eta-1}{e+1} + \frac{1-e^{-1}}{2(e^{\eta}-1)} 
\geq \sqrt{T \cdot \frac{1-e^{-1}}{e+1}} = \Omega(\sqrt{T}). \]
To sum up, we have $\max\{r(T), r'(T)\} \geq \Omega(\sqrt{T})$.
\end{proof}

\end{document}